\documentclass[11pt]{article}
\usepackage{rotating,amsmath,amssymb,amsfonts,amsthm,longtable}
\usepackage{epsfig,mathrsfs,natbib,color,graphics,bm,dsfont,subfigure}
\usepackage{multirow}
\def\boxit#1{\vbox{\hrule\hbox{\vrule\kern6pt \vbox{\kern6pt#1\kern5pt}
\kern6pt\vrule}\hrule}}

\newcommand{\bx}{{\boldsymbol x}}
\newcommand{\bz}{{\boldsymbol z}}
\newcommand{\bw}{{\boldsymbol w}}

\newcommand{\bC}{{\boldsymbol C}}

\newcommand{\bR}{{\boldsymbol R}}
\newcommand{\bs}{{\boldsymbol s}}
\newcommand{\bS}{{\boldsymbol S}}
\newcommand{\bX}{{\boldsymbol X}}
\newcommand{\bY}{{\boldsymbol Y}}

\newcommand{\tbx}{\tilde{\boldsymbol{x}}}
\newcommand{\bXmis}{{\boldsymbol X}^{{\rm mis}}}
\newcommand{\bXobs}{{\boldsymbol X}^{{\rm obs}}}
\newcommand{\tbXmis}{{\tilde{\boldsymbol X}}^{{\rm  mis}}}
\newcommand{\tbxmis}{{\tilde{\boldsymbol x}}^{{\rm  mis}}}

\newcommand{\txmis}{{\tilde{x}}^{{\rm  mis}}}
\newcommand{\Xmis}{X^{{\rm  mis}}}
\newcommand{\Xobs}{X^{{\rm  obs}}}
\newcommand{\bxmis}{{\boldsymbol x}^{{\rm  mis}}}
\newcommand{\bxobs}{{\boldsymbol x}^{{\rm  obs}}}

\newcommand{\xobs}{x^{{\rm  obs}}}
\newcommand{\mN}{\mathbb{N}}
\newcommand{\mR}{\mathbb{R}}

\newcommand{\mX}{{\cal X}}

\newcommand{\tM}{M}

\newcommand{\bbeta}{{\boldsymbol \beta}}
\newcommand{\btheta}{{\boldsymbol \theta}}

\newcommand{\bepsilon}{{\boldsymbol \epsilon}}
\newcommand{\bmu}{{\boldsymbol \mu}}

\newcommand{\bSigma}{{\boldsymbol \Sigma}}
\newcommand{\blambda}{{\boldsymbol \lambda}}
\newcommand{\bgamma}{{\boldsymbol \gamma}}

\newtheorem{theorem}{Theorem}[]
\newtheorem{lemma}{Lemma}[]

\newtheorem{corollary}{Corollary}[]

\def\ANNALS{{\it Annals of Statistics}}

\def\JRSSB{{\it Journal of the Royal Statistical Society, Series B}}

\def\JASA{{\it Journal of the American Statistical Association}}

\def\JASA{{\it Journal of the American Statistical Association}}
\def\ANNALS{{\it Annals of Statistics}}

\def\JRSSB{{\it Journal of the Royal Statistical Society, Series B}}

\def\ANNALS{{\it Annals of Statistics}}

\def\JRSSB{{\it Journal of the Royal Statistical Society, Series B}}

\def\JASA{{\it Journal of the American Statistical Association}}

\def\PNAS{{\it Proceedings of the National Academy of Sciences USA}}

\def\JMLR{{\it Journal of Machine Learning Research}}
\begin{document}
\bibliographystyle{asa}

\title{An Imputation-Consistency Algorithm for High-Dimensional Missing Data Problems and Beyond}

\author{} 

 \author{Faming Liang, Bochao Jia, Jingnan Xue, Qizhai Li, Ye Luo\thanks{
  F. Liang is with Department of Statistics, Purdue University, West Lafayette, IN 47907,
  email: fmliang@purdue.edu; B. Jia is with Department of Biostatistics, 
  University of Florida, Gainesville, FL 32611;
  J. Xue is with Department of Statistics, 
  Texas A\&M University, College Station, TX 77843. 
  Q. Li is with Academy of Mathematics and Systems Science,
  Chinese Academy of Sciences, Beijing 100864, China.  
  Y. Luo is with Department of Economics, University of Florida, Gainesville, FL 32611.
  }
  }

\maketitle

\begin{abstract}

 Missing data are frequently encountered in high-dimensional problems, but they are usually
 difficult to deal with using standard algorithms, such as the expectation-maximization (EM) algorithm and its variants. 
 To tackle this difficulty, some problem-specific algorithms  have been developed 
 in the literature, but there still lacks a general algorithm. This work is to fill the gap: 
 we propose a general algorithm for high-dimensional missing data problems. The proposed 
 algorithm works by iterating between an imputation step and a consistency step. 
 At the imputation step, the missing data are imputed conditional on the observed data 
 and the current estimate of parameters; and at the consistency step, a consistent estimate 
 is found for the minimizer of a Kullback-Leibler divergence defined on the pseudo-complete data. 
 For high dimensional problems, the consistent estimate can 
 be found under sparsity constraints.
 The consistency of the averaged estimate for the true parameter can be established 
 under quite general conditions. The proposed algorithm is illustrated using 
 high-dimensional Gaussian graphical models, high-dimensional variable selection, 
 and a random coefficient model. 
% The proposed algorithm has strong implications for big 
% data computing: Based on it, we propose a general strategy to improve Bayesian computation 
% for big data problems. The proposed algorithm also facilitates data integration 
% from multiple sources, which plays an important role in big data analysis. 
 
\vspace{1mm}


\end{abstract}

% \newpage 

{\centering \section{Introduction}}

 Missing data are frequently encountered in both low and high-dimensional data, 
 where low and high refer to that the number of variables is smaller or larger than the sample size, respectively.  
 For example, the microarray data is usually considered as high-dimensional, where the number of genes 
 can be much larger than the number of samples. 
 Missing values can appear in microarray data due to various factors such as scratches on slides,
 spotting problems, experimental errors, etc. In some microarray experiments,
 missing values can occur for more than 90\% of the genes (Ouyang et al., 2004). 
 Simply deleting the samples or genes for which missing values occur can lead to  
 a significant loss of information of the data. How to deal with missing data 
 has been a long-standing problem in statistics.

 For low-dimensional problems, the missing data can be dealt with using  
 the EM algorithm (Dempster et al., 1977) or its variants. 
 Let $\bXobs=(\Xobs_1,\Xobs_2,\ldots,\Xobs_n)$ denote the 
 observed incomplete data, where $n$ denotes the sample size. Let $\bXmis=(\Xmis_1, \Xmis_2,\ldots, \Xmis_n)$ 
 denote the missing data, and let $\bX=(\bXobs,\bXmis)$ denote the complete data. Let $\btheta$ denote 
 the vector of unknown parameters, and let $f(\bX|\btheta)$ denote the likelihood function 
 of the complete data. Then the maximum likelihood estimate (MLE) of $\btheta$ can be determined by 
 maximizing the marginal likelihood of the observed data, 
 \[
  f(\bXobs|\btheta)=\int f(\bXobs,\bxmis|\btheta) h(\bxmis|\btheta, \bXobs) d\bxmis, 
 \]
 where $h(\bxmis|\btheta,\bXobs)$ denotes the predictive density of the missing data.  
 The EM algorithm seeks to maximize the marginal likelihood function  by iterating between  
 the following two steps:
 \begin{itemize} 
 \item {\bf E-step}: {\it Calculate the expected value of the log-likelihood function with respect to 
  the predictive distribution of the missing data given the current estimate $\btheta^{(t)}$, i.e., }
  \[
  Q(\btheta|\btheta^{(t)})=\int \log f(\bXobs,\bxmis|\btheta) h(\bxmis|\btheta^{(t)},\bXobs) d\bxmis.
  \]
  \item {\bf M-step}: {\it Find a value of $\btheta$ that maximizes the quantity  $Q(\btheta|\btheta^{(t)})$, i.e., set }
   \[
   \btheta^{(t+1)}=\arg\max_{\btheta} Q(\btheta|\btheta^{(t)}).
   \]
 \end{itemize} 
 Dempster et al. (1977) showed that the marginal likelihood value increases with each iteration and,
 under fairly general conditions, it converges to a local or global maximum of the 
 marginal likelihood. A rigorous study for the convergence is given by Wu (1983). 

 Both the $E$ and $M$-steps of the algorithm can be rather complicated or even intractable.  
 Meng and Rubin (1993) found that in many cases, the $M$-step is relatively simple when conditioned on 
 some function of the parameters under estimation. Motivated by this observation, they introduced 
 the expectation-conditional maximization (ECM) algorithm, which is to replace the M-step by a number 
 of computationally simpler conditional maximization steps. Later, the EM algorithm was 
 further speeded up by some other variants, such as the ECME algorithm (Liu and Rubin, 1994, 
 He and Liu, 2012) and the PX-EM algorithm (Liu et al., 1998).  
 When the E-step is analytically intractable, Wei and Tanner (1990) introduced the 
 Monte Carlo EM algorithm, which is to simulate multiple missing values from the 
 predictive distribution $h(\bxmis|\btheta^{(t)},\bXobs)$ at the $(t+1)th$ iteration, 
 and then maximize the approximate conditional expectation of the complete-data log-likelihood 
 \[
  \widehat{Q}(\btheta|\btheta^{(t)}) = \frac{1}{m} \sum_{j=1}^m \log f(\bXobs,\bXmis_j|\btheta),
 \]
  which converges to $ Q(\btheta|\btheta^{(t)})$ as $m\to \infty$, where 
  $\bXmis_1,\ldots, \bXmis_m$ denote the missing values simulated from $h(\bxmis|\btheta^{(t)}$, $\bXobs)$. 
  When the dimension of $\bXmis$ is 
  high, the Monte Carlo approximation can be rather expensive. 
  An alternative algorithm to deal with the intractable E-step is the 
  stochastic EM (SEM) algorithm (Celeux and Diebolt, 1985). In this algorithm, the E-step is replaced 
  by an imputation step, where the missing data are imputed with plausible values conditioned on the observed 
  data and the current parameter estimate. At the M-step, the parameters are estimated by    
  maximizing the likelihood function of the pseudo-complete data. 
  Unlike the deterministic EM algorithm, the imputation-step 
  and M-step of the SEM algorithm generate a Markov chain which converges to a stationary distribution 
  whose mean is close to the MLE and whose variance reflects the information loss due to 
  the missing data (Nielsen, 2000). 

  Although EM and its variants work well for low-dimensional problems,   
  see McLachlan and Krishnan (2008) for an overview, 
  they essentially fail for high-dimensional problems. 
  For the latter, the MLE 
  can be non-unique or inconsistent. To address this issue, 
  some problem-specific algorithms have been proposed, see e.g., misgLasso (St\"adler and B\"uhlmann, 2012),
  misPALasso (St\"adler et al., 2014), and matrix completion algorithms (Cai et al. 2010;
  Mazumder et al., 2010).  
  MisgLasso is specifically designed for estimating Gaussian graphical models
  in presence of missing data. Similar to misgLasso, MisPALasso also deals with
  multivariate Gaussian data in presence of missing data.
  The matrix completion algorithm deals with large incomplete matrices, which is to learn a low-rank 
  approximation for a large-scale matrix with missing entries.  
  However, there still lacks a general algorithm for high-dimensional missing data problems. 

  This work is to fill the gap: we propose a general algorithm for 
  dealing with high-dimensional missing data problems. 
  The proposed algorithm consists of two steps, an imputation step and a consistency step, and 
  is therefore called an imputation-consistency (IC) algorithm.  
  The imputation step is to impute the missing data with plausible values 
  conditioned on the observed data and the current estimate of the parameters.  
  The consistency step is to find a consistent estimate for the minimizer of a 
  Kullback-Leibler divergence defined on the  
  pseudo-complete data. For high dimensional problems, the consistent estimate is suggested
  to be found under sparsity constraints. 
  Like the SEM algorithm, the IC algorithm generates a Markov chain which converges to a 
  stationary distribution. Under mild conditions, we show that the mean of the stationary
  distribution converges to the true value of the parameters in probability as the sample 
  size becomes large. % The IC algorithm is general, which, in principle, can be 
  % applied to any missing data problems, regardless the dimension and distribution of the data.  
  For low-dimensional problems, the SEM algorithm can be viewed as a special case of the IC algorithm.
  The IC algorithm has strong implications for big data computing: Based on it,
  we propose a general strategy to improve Bayesian computation for big data. 
  The IC algorithm also facilitates data integration from multiple sources,
  which plays an important role in big data analysis. 
  A R package accompanying this paper is currently available at 
  {\it http://www.stat.purdue.edu/$\sim$fmliang} and later will be 
  distributed to the public via CRAN upon the acceptance of the paper.

  The remainder of this paper is organized as follows. Section 2 describes the IC 
  algorithm with the theoretical development deferred to the 
  Appendix. Section 3 applies the IC algorithm to  
  high-dimensional Gaussian graphical models.  Section 4 applies the IC algorithm
  to high-dimensional variable selection. Section 5 applies
  the IC algorithm to a random coefficient model and discusses its potential use 
  for big data problems.
  Section 6 concludes the paper with a brief discussion. 
 % In particular, we discuss how 
 %  the IC algorithm facilitates data integration from multiple sources.
 
{\centering \section{The Imputation-Consistency Algorithm}} 
 
\subsection{The IC Algorithm} 
 
 Let $X_{1},\dots,X_{n}$ denote a random sample drawn from the distribution $f(x|\btheta)$ (also denoted by 
 $f_{\btheta}(x)$ depending on convenience),  where $\btheta$ is a vector of parameters.   
 Let $X_i=(\Xobs_i,\Xmis_i)$, $i=1,\ldots, n$, where $\Xobs_i$ is observed and $\Xmis_i$ is missed. 
 Let $\bX=(X_1,\ldots, X_n)$, $\bXobs=(\Xobs_1,\ldots,\Xobs_n)$ and $\bXmis=(\Xmis_1,\ldots, \Xmis_n)$. 
 To indicate the dependence of the dimension of $\btheta$ on
 the sample size $n$, we also write $\btheta$ as $\btheta_n$ and denote by $\btheta_n^{(t)}$ 
 the estimate of $\btheta$ obtained at the $t^{th}$ iteration of the IC algorithm. 
 The IC algorithm works by 
 starting with an initial guess $\btheta_n^{(0)}$ and then iterating between 
 the imputation and consistency steps: 
 
\begin{itemize}
\item {\bf I-step}: {\it Draw $\tbXmis$ from the predictive distribution 
  $h(\bxmis|\bXobs, \btheta_n^{(t)})$ given $\bXobs$ and the current estimate $\btheta_n^{(t)}$.}  

\item {\bf C-step}: {\it Based on the pseudo-complete data $\tilde{\bX}=(\bXobs,\tbXmis)$, find 
   an updated estimate $\btheta_n^{(t+1)}$ which forms a consistent estimate of 
   \begin{equation} \label{Cequation}
   \btheta_*^{(t)}=\arg\max_{\btheta} E_{\btheta_n^{(t)}} \log f_{\btheta}(\tilde{\bx}),
   \end{equation} 
   where $E_{\btheta_n^{(t)}} \log f_{\btheta}(\tilde{\bx})
   = \int \log(f(\bxobs, \tbxmis|\btheta))f(\bxobs|\btheta^*) h(\tbxmis|\bxobs,\btheta_n^{(t)}) d\bxobs d \tbxmis$, 
   $\btheta^*$ denotes the true value of the parameters, and $f(\bxobs|\btheta^*)$ 
   denotes the marginal density function of $\bxobs$. }
 \end{itemize}
 
 To find a consistent estimate of $\btheta_*^{(t)}$, which is the minimizer of the Kullback-Leibler 
 divergence from $f(\tilde{\bx}|\btheta)$ to the joint density $f(\bxobs|\btheta^*) h(\tbxmis|\bxobs,\btheta_n^{(t)})$, 
 sparsity constraints can be imposed on $\btheta$ for high-dimensional problems. 
 In general, we have two ways. The first way is via regularization methods. 
 Corollary \ref{cor0A} in the Appendix shows that the regularization methods 
 can be employed here to find consistent estimates for $\btheta_*^{(t)}$'s  
 with appropriate penalty functions.  
 For regularization methods, we recommend to use the same penalty functions as 
 they would use if there are no missing data.
 The second way is via sure screening-based methods, which are to first reduce the 
 space of $\btheta_*^{(t)}$ to a low-dimensional subspace and then find a 
 consistent estimate of $\btheta_*^{(t)}$ in the low-dimensional subspace using a 
 conventional statistical method, such as maximum likelihood, moment estimation or even regularization. 
 In the Appendix, we point out that the sure screening-based methods can be viewed as 
 a subclass of regularization methods, for which the solutions in the low-dimensional 
 subspace receives a zero penalty and those outside the subspace receives a penalty of $\infty$. 
 Such a binary-type penalty function satisfies the condition (C1) we imposed on  
 regularization methods.  Other than the regularization and sure screening-based methods, 
 we justify in Corollary \ref{cor0} and Remark (R3) the use of general 
 consistent estimation procedures in the IC algorithm, provided that 
 the resulting estimates are accurate enough at each iteration $t$. 
 For low-dimensional problems, the consistent estimator of $\btheta_*^{(t)}$ can be obtained by
 maximizing the pseudo-complete likelihood function. In this sense,
 the SEM algorithm can be viewed as a special case of the IC algorithm.

 It is easy to see that by simulating new independent missing values 
 at each iteration, the sequence of estimates, $\{\btheta_n^{(t)} \}$, 
 forms a time-homogeneous Markov chain. Also, the imputed values  
 at different iterations form a Markov chain. 
 The two Markov chains are interleaved and share many properties, such as 
 irreducibility, aperiodicity and ergodicity. Refer to 
 Nielsen (2000) for more discussions on this issue. 
 In Theorem \ref{them1} and Theorem \ref{them2} of the Appendix, we prove that  
 the Markov chain $\{\btheta_n^{(t)} \}$ has a stationary distribution and, furthermore, 
 the mean of the stationary distribution forms a consistent estimate of $\btheta^*$.
 Like for other Markov chains, a good initial value will accelerate 
 the convergence of the simulation.
 There are many different ways to specify initial values for the IC algorithm.   
 In most examples of this paper, we started the 
 simulation with an I-step, where all missing values are filled by the 
 median of the variable. 
 This method is simple and usually works when the missing rate is not high. 
 However, when we perceive that such a constant filling method does not 
 work well, we may start the simulation with a C-step. 
 In this case, the initial estimate $\btheta_n^{(0)}$ may be obtained 
 based on the complete samples (i.e., those without missing information) only. 
 
 We note that many of the assumptions we made for proving the convergence of the IC algorithm 
 are quite regular. For example, we assumed that $\log f_{\btheta}(\tilde{x})$ is 
 a continuous function of $\btheta$ for each $\tilde{x} \in \mX$ and a measurable function 
 of $\tilde{x}$ for each $\btheta$. Since 
  we aim to address the missing data issue for a wide range of problems
  and it is hard to specify the structure of each problem, we incorporate the assumptions
  about the parameters and problem structures into a metric entropy condition, see 
  condition (A2). As discussed in Remark (R1) of
 the Appendix, this condition allows $p$ 
 to grow with $n$ at a polynomial rate $O(n^\gamma)$ for some 
 constant $0<\gamma<\infty$, and allows the number of nonzero elements 
  in $\btheta$ to grow with $n$ at a rate of $O(n^{\alpha})$ for some 
  $0<\alpha<1/2$. These rates seem a little more restrictive than the 
 exponential rate, i.e., $\log(p)=n^{b}$ for some constant $0<b<1$, 
 seeking for in the literature of high-dimensional regression. However, 
 more or less, they are just some technical conditions.  Moreover, our theory 
 is more general and can be applied to many other problems. 
 Note that the metric entropy condition has often been used in studying the minimax
 rate of estimation under the high-dimensional scenario, see e.g. Raskutti et al. (2011).

 Regarding conditions on missing data, we note that 
 the IC algorithm essentially works with any missing
 data mechanism, as long as the predictive distribution  $h(\bxmis|\bXobs, \btheta_n^{(t)})$
 is available, well behaved, and unchanged with the sample size $n$. Our
 current theory rules out the case that the missing data mechanism
 changes as the sample size increases, e.g., the missing rate increases due to increased wear and tear
 on measurement instruments or fatigue among data subjects measured later in the study. 
 Our condition (A3) constrains the behavior of $h(\bxmis|\bXobs, \btheta_n^{(t)})$ via 
 some moment conditions on the log-likelihood function of the pseudo-complete data. 
 It implies that a high missing rate may hurt the performance of the method.  

\subsection{An Extension of the IC Algorithm}
 
 Like the EM algorithm, the IC algorithm is attractive only when the consistent estimate of $\btheta_*^{(t)}$
 can be easily obtained at each C-step. 
 We found that for many problems, similar to the ECM algorithm (Meng and Rubin, 1993), 
 the consistent estimate of $\btheta_*^{(t)}$ can be easily obtained with a number of conditional consistency steps. 
 That is, we can partition the parameter $\btheta$ into a number of blocks and then find 
 the consistent estimator for each block conditioned on the current estimates of other 
 blocks. Note that 
 for many problems, e.g., the examples studied in Sections 4 and 5, the partitioning 
 of $\btheta$ is natural. 
 
 Suppose that $\btheta=(\btheta^{(1)}, \ldots, \btheta^{(k)})$ has been partitioned into 
 $k$ blocks. The imputation-conditional consistency (ICC) algorithm can be described as 
  follows:
 
 \begin{itemize}
 \item {\bf I-step}. Draw $\tbXmis$ from the conditional distribution
  $h(\bxmis|\bXobs, \btheta_n^{(t,1)}, \ldots, \btheta_n^{(t,k)})$ 
  given $\bXobs$ and the current estimate $\btheta_n^{(t)}=(\btheta_n^{(t,1)},\ldots, \btheta_n^{(n,k)})$.

\item {\bf CC-step}. Based on the pseudo-complete data $\tilde{\bX}=(\bXobs,\tbXmis)$, 
   do the following:
   \begin{itemize}
    \item[(1)] Conditioned on $(\btheta_n^{(t,2)}, \ldots, \btheta_n^{(t,k)})$, find $\btheta_n^{(t+1,1)}$ 
     which forms a consistent estimate of 
         \[
          \btheta_*^{(t,1)}=\arg\max_{\btheta^{(t,1)'}} E_{\btheta_n^{(t,1)},\ldots, \btheta_n^{(t,k)}} 
          \log f(\tilde{\bx}| \btheta_n^{(t,1)'}, \btheta_n^{(t,2)}, \ldots, \btheta_n^{(t,k)} ),
         \]
    where the expectation is taken with respect to the joint distribution function of $\tilde{\bx}=(\bxobs,\bxmis)$
    and the subscript of $E$ gives the current estimate of $\btheta$. 

    \item[(2)] Conditioned on $(\btheta_n^{(t+1,1)}, \btheta_n^{(t,3)}, \ldots, \btheta_n^{(t,k)})$, find $\btheta_n^{(t+1,2)}$
     which forms a consistent estimate of
         \[
          \btheta_*^{(t,2)}=\arg\max_{\btheta^{(t,2)'}} E_{\btheta_n^{(t+1,1)}, \btheta_n^{(t,2)}, \btheta_n^{(t,3)},
           \ldots, \btheta_n^{(t,k)}} 
          \log f(\tilde{\bx}| \btheta_n^{(t+1,1)}, \btheta_n^{(t,2)'}, \btheta_n^{(t,3)}, \ldots, \btheta_n^{(t,k)} ).
         \]

    \item[] $\ldots\ldots$
    
    \item[(k)] Conditioned on $(\btheta_n^{(t+1,1)}, \ldots, \btheta_n^{(t,k-1)})$, find $\btheta_n^{(t+1,k)}$
     which forms a consistent estimate of
         \[
          \btheta_*^{(t,k)}=\arg\max_{\btheta^{(t,k)'}} E_{\btheta_n^{(t+1,1)}, 
           \ldots, \btheta_n^{(t+1,k-1)}, \btheta_n^{(t,k)}}
          \log f(\tilde{\bx}| \btheta_n^{(t+1,1)}, \ldots, \btheta_n^{(t+1,k-1)}, \btheta_n^{(t,k)'} ).
         \]
  \end{itemize}
 \end{itemize}
 
 It is easy to see that the sequence $\{(\btheta_n^{(t,1)}, \ldots, \btheta_n^{(t,k)})\}$  
 forms a Markov chain. 
 The convergence of the Markov chain can be studied under similar conditions as the IC algorithm. 
 In Theorem \ref{them3} and Theorem \ref{them4} (see Appendix), we prove that
 the Markov chain $\{\btheta_n^{(t)} \}$ has a stationary distribution and 
 the mean of the stationary distribution forms a consistent estimate of $\btheta^*$.

{\centering \section{Learning High-Dimensional Gaussian Graphical Models in Presence of Missing Data}}
 
 Gaussian graphical models (GGMs) have often been used in learning gene 
 regulatory networks from microarray data, see e.g., 
 Dobra et al. (2004) and Friedman et al. (2008). 
 As mentioned in the Introduction, 
 missing values can appear in microarray data due to many factors. 
 To deal with missing values in microarray data, many imputation methods, 
 such as % $k$-nearest neighbor imputation (Troyanskaya et al., 2001),
 single value decomposition (SVD) imputation (Troyanskaya et al., 2001),
 least-square imputation (Bo et al., 2004), and Bayesian principal component 
 analysis (BPCA) imputation (Oba et al., 2003), have been proposed. 
 Since these methods impute the missing values 
 independent of the models under consideration, they are often ineffective.
 Moreover, the statistical inference based on the ``one-time'' imputed 
 data is potentially biased, because the uncertainty of the missing values cannot 
 be properly accounted for.  
 In this section, we apply the IC algorithm to handle missing values  
 for microarray data. The IC algorithm iteratively impute missing 
 values based on the updated parameter estimate. Therefore, it overcomes 
 the weakness of the ``one-time'' imputation methods, and improves 
 accuracy of statistical inference. 

 Let $\bX=(\bx_1, \ldots, \bx_n)^T$ denote a microarray dataset of $n$ samples and
 $p$ genes, where $\bx_i$  
 is assumed to follow a multivariate Gaussian distribution $N_p(\bmu, \bSigma)$.  
 According to the theory of GGMs, estimation of the GGM is  
 equivalent to identify non-zero elements of the concentration matrix 
 (i.e., the inverse of the covariance matrix $\bSigma$) or to identify 
  non-zero partial correlation coefficients for different pairs of genes.  
 During the recent years, a couple of methods have been proposed to estimate 
 high-dimensional GGMs, e.g., graphical Lasso 
 (Yuan and Lin, 2007; Friedman et al., 2008), node-wise regression (Meinshausen and B\"uhlmann, 2006), 
 and $\psi$-learning (Liang et al., 2015). 
 % Graphical Lasso is to estimate the concentration matrix using the regularization method with a $L_1$-penalty.
 % Node-wise regression is developed based on the relationship between 
 % the partial correlation coefficients and regression coefficients, and it is to 
 % select the non-zero regression coefficients for each gene regressed with all other genes. 
 % The $\psi$-learning algorithm is developed based on an equivalent measure of the partial correlation 
 % coefficient, which is evaluated  with a reduced conditional set of genes and 
 % thus feasible for high-dimensional problems. 
 However, none of the methods can be directly applied in presence of missing data. 

\subsection{The IC Algorithm}

 To apply the IC algorithm to learn GGMs in presence of missing data, 
 we choose the $\psi$-learning algorithm as the consistent estimation procedure used in the C-step.  
 For GGMs, $\btheta$ corresponds to the concentration matrix, which 
 can be uniquely determined from the network structure using the algorithm given in Hastie et al. (2009, p.634). 
 Under mild conditions, Liang et al. (2015) showed that the $\psi$-learning algorithm provides a consistent 
 estimator for Gaussian graphical networks. 
 Refer to the Supplementary Material for a brief review of the 
 algorithm. As mentioned in the Appendix, the $\psi$-learning algorithm  
 belongs to the class of sure-screening-based methods, which are to first
 reduce the dimension of the solution space via correlation screening and then conduct GGM estimation
 via covariance selection (Dempster, 1972) which, by nature, produces a
 maximum likelihood estimate. 
 As mentioned in Section 2.1, such a sure-screening-based method can be used in IC simulations.
 Other than the $\psi$-learning algorithm, 
 node-wise regression and graphical Lasso can also be used, which
 both belong to the class of regularization methods and are 
 consistent in Gaussian graphical network estimation.  

 The Gaussian graphical network specifies the dependence between different 
 genes, according to which the missing values can be imputed. 
 For convenience, we let $A=(a_{jk})$ denote the adjacency matrix of a Gaussian graphical network, where 
 $a_{jk}=1$ if an edge exists between node $j$ and node $k$ and 0 otherwise. 
 For microarray data, a node corresponds to a gene. 
 Let $x_{ij}$ denote a missing entry, and let $\omega(j)=\{k: a_{jk}=1 \}$ denote the 
 neighborhood of node $j$. According to the faithfulness property of GGMs, 
 conditional on the neighboring genes in $\omega(j)$, gene $j$ is independent of all 
 other genes. Therefore, $x_{ij}$ can be imputed conditional on the 
 expression values of the neighboring genes. Mathematically, we have 
\begin{equation}\label{10}
\left(\begin{array}{c} 
x_{ij}\\
\bx_{i\omega}\\
\end{array}\right) \sim N\left(\left(\begin{array}{c} 
 \mu_j \\
 \bmu_{\omega} \\
\end{array}\right),\left(\begin{array}{cc} 
\sigma_j^2 & \bSigma_{j \omega}\\
\bSigma_{j\omega}^T & \bSigma_{\omega\omega}\\
\end{array}\right)\right),
\end{equation}
where $\bx_{i\omega}=\{x_{ik}: k\in \omega(j)\}$, and $\mu_j$, $\bmu_{\omega}$, 
 $\sigma_j^2$, $\bSigma_{j \omega}$ and $\bSigma_{\omega\omega}$ denote the 
 corresponding mean and variance components. 
 The mean and variance of $x_{ij}$ conditional 
 on $\bx_{i\omega}$ is thus given by 
 \begin{equation} \label{condeq1}
% \begin{split} 
  \mu_{ij|\omega}  = \mu_j+\bSigma_{j\omega} \bSigma_{\omega\omega}^{-1} (\bx_{i\omega}-\bmu_{\omega}), \quad
  \sigma_{ij|\omega} = \sigma_j^2- \bSigma_{j\omega} \bSigma_{\omega\omega}^{-1} \bSigma_{j\omega}^T. 
% \end{split}
 \end{equation}
 As shown in Liang et al. (2015), for each gene, the neighborhood size can be 
 upper bounded by $\lceil n/log(n)\rceil$, where $\lceil z\rceil$ 
 denotes the smallest integer not smaller than $z$. Hence, in practice, 
 $\sigma_j^2$, $\bSigma_{j\omega}$ and $\bSigma_{\omega\omega}$ can be directly estimated from 
 the data. Let $\bs_j^2$, $\bS_{j\omega}$, $\bS_{\omega\omega}$, $\bar{x}_j$ and $\bar{\bx}_{\omega}$ 
 denote the respective sample estimates of $\sigma_j^2$, $\bSigma_{j\omega}$,  $\bSigma_{\omega\omega}$,
 $\mu_j$ and $\bmu_{\omega}$. Then, at each iteration, $x_{ij}$ can be imputed by sampling from the 
 distribution
 \begin{equation} \label{imputeq1}
  X_{ij|\omega} \sim N(\bar{x}_j+\bS_{j\omega} \bS_{\omega\omega}^{-1} (\bx_{i\omega}-\bar{\bx}_{\omega}), 
     \bs_j^2- \bS_{j\omega} \bS_{\omega\omega}^{-1} \bS_{j\omega}^T).
 \end{equation} 
 In this way, exact evaluation of the concentration matrix can be skipped. 
 In summary, we have the following algorithm for learning GGMs in presence of missing data: 
 
 {\it
 \begin{itemize} 
 \item (Initialization) Fill each missing entry by the median of 
        the corresponding variable, and then iterates between the C- and I-steps.
 
 \item (C-step) Apply the $\psi$-learning algorithm to learn the structure of the 
       Gaussian graphical network.

 \item (I-step) Impute missing values according to (\ref{imputeq1}) based on the 
        network learned in the C-step. 
 \end{itemize}
 }
 
 This algorithm outputs a series of Gaussian graphical networks. To integrate/average these 
 networks into a single network, we adopt the $\psi$-score averaging 
 approach suggested by Liang et al. (2015). 
 Let $(\psi_{ij}^{(t)})$ denote the 
 $\psi$-scores at iteration $t$, where are obtained 
 from $\psi$-partial correlation coefficients 
 via Fisher's transformation. 
 Let $\bar{\psi}_{ij}=\sum_{t=1}^T \psi_{ij}^{(t)}/T$, $i,j=1,2,\ldots,p$ and $i\ne j$, 
  denote the averaged $\psi$-score for gene $i$ and gene $j$. 
 Then the averaged network can be obtained by applying a multiple hypothesis approach  
 to threshold the averaged $\psi$-scores; if an averaged $\psi$-score is greater than the threshold value, 
 we set the corresponding element of the adjacency matrix to 1 and 0 otherwise.    
 The multiple hypothesis test can be done using the method of Liang and Zhang (2008), 
 which can be viewed as a generalized empirical Bayesian method (Efron, 2004). 
 The significance level of the multiple hypothesis test can be specified in terms of 
 Storey's $q$-value (Storey, 2002). In this paper, we set it to $0.05$.

\subsection{A Simulated Example}

 We consider an autoregressive process of order two with the concentration matrix given by 
 \begin{equation}\label{plugin}
        C_{i,j}=\left\{\begin{array}{ll}
                     0.5,&\textrm{if $\left| j-i \right|=1, i=2,...,(p-1),$}\\
                     0.25,&\textrm{if $\left| j-i \right|=2, i=3,...,(p-2),$}\\
                     1,&\textrm{if $i=j, i=1,...,p,$}\\
                    0,&\textrm{otherwise.}
                \end{array}\right.
\end{equation}
 This example has been used by multiple authors, e.g., Yuan and Lin (2007), Mazumder and Hastie (2012), 
 and Liang et al. (2015) to illustrate different GGM methods. In this paper, we generated multiple datasets 
 with $n=200$ and different values of $p$=100, 200, 300 and 400. 
 For each combination of $(n,p)$, we generated 10 datasets independently; 
 and for each dataset, we randomly deleted 10\% of  
 the observations as missing values. To evaluate the performance of the IC algorithm, the precision-recall curves 
 were drawn by varying the threshold value of $\psi$-scores, where
 the precision is the fraction of true edges 
 among the retrieved edges, and the recall is the fraction of true edges that have been retrieved 
 over the total amount of true edges. 

\begin{figure}[htbp]
\begin{center}
\begin{tabular}{c}
%  (a) $p=100$ & (b) $p=400$ \\
% \epsfig{figure=response/n200p100.eps,height=3.0in,width=2.0in,angle=270} & 
 \epsfig{figure=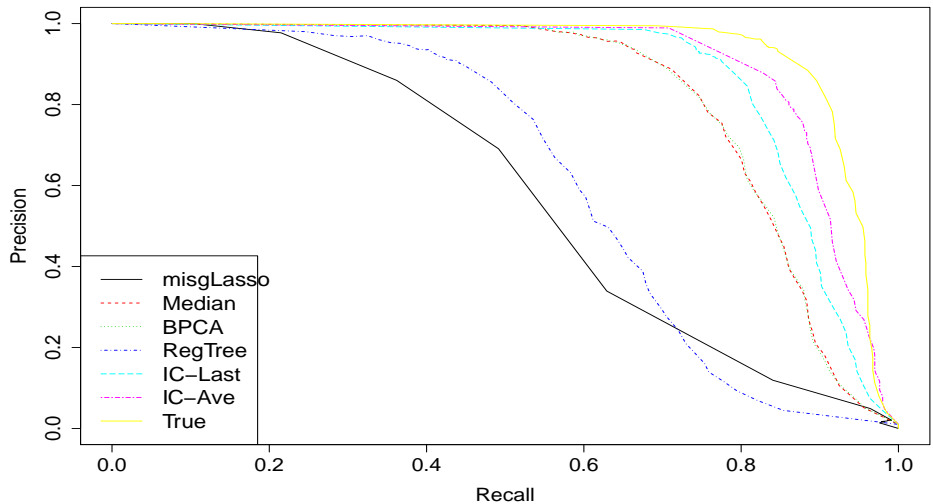,height=5.0in,width=3.0in,angle=270}
\end{tabular}
\caption{
 Precision-recall curves resulted from different imputation methods for one simulated dataset with $p=400$:
``True'' refers to the curve obtained with complete data;
 ``misgLasso'' refers to the curve produced by the misgLasso algorithm;
 ``IC-Ave'' and ``IC-Last'' refer to
 the curves obtained with the $\psi$-scores generated in the last IC iteration and
 averaged over last 20 IC iterations, respectively; and
 ``Median'', ``BPCA'' and ``RegTree'' refer to the curves obtained
 with missing values imputed by the median filling, BPCA and regression tree methods, respectively. }
\label{PRcurveEx1} %% label for entire figure
\end{center}
\end{figure}

 For each dataset, the IC algorithm was run for 50 iterations. 
 Figure \ref{PRcurveEx1} shows the resulting precision-recall curves for one 
 dataset with $p=400$. 
 For comparison, Figure  \ref{PRcurveEx1} also includes the precision-recall curves produced by
 misgLasso and those produced by the $\psi$-learning algorithm with missing values 
 imputed by the median filling, BPCA(Oba et al., 2003), and regression tree 
 (Buuren and Groothuis-Oudshoorn, 2011) methods. 
 The regression tree method has been implemented 
 in the R package {\it MICE} and was applied to this example under its default setting. 
 %As mentioned previously, BPCA is a method specially designed for imputing 
 %missing values of the microarray data. Under the Bayesian framework, this method estimates 
 %the probabilistic principal components of the data using 
 %an EM-like iterative algorithm, and then imputes the missing values based on the principal component regression.  
 %For both the ``Median'' and BPCA methods, the GGM will be learned using the 
 %$\psi$-learning algorithm after the missing values are imputed. 
 The misgLasso algorithm is a combination of the gLasso and EM algorithms, which is to 
 integrate out the missing data as in the EM algorithm (see e.g., St\"adler and B\"uhlmann, 2012) 
 and then learn the GGM using the gLasso algorithm.   
 The misgLasso algorithm has been implemented in the R package {\it spaceExt} (He, 2011).
 Refer to Figure 1 of the Supplementary Material for
 the curves with other values of $p$.

\begin{table}[htbp]
\tabcolsep=3pt\fontsize{7}{9}
\begin{center}
\caption{Average areas (over 10 datasets) under the Precision-Recall curves resulted from
  different imputation methods, where the number in the parentheses
  denotes the standard deviation of the average. }
\label{AUC}
\vspace{2mm}
\begin{tabular}{ccccccccc} \hline
 p& misgLasso & Median & BPCA & RegTree &  IC-Last & IC-Ave & True\\ \hline
   100 & 0.678(0.006)& 0.882(0.007)& 0.874(0.006)& 0.817(0.005)& 0.877(0.007)& 0.904(0.006)& 0.949(0.006) \\\hline
 200 & 0.633(0.005)& 0.856(0.004)& 0.855(0.004)& 0.442(0.004)& 0.887(0.004)& 0.902(0.003)& 0.941(0.002) \\\hline
300 &0.599(0.004)& 0.830(0.003)& 0.833(0.003)& 0.574(0.003)& 0.869(0.003)& 0.901(0.002)& 0.936(0.002) \\\hline
400 & 0.580(0.003)& 0.824(0.003)& 0.824(0.003)& 0.620(0.003)& 0.868(0.003)& 0.900(0.002)& 0.932(0.001) \\\hline
\end{tabular}
\end{center}
\end{table}

 Table \ref{AUC} compares the averaged areas (over 10 datasets) under the precision-recall curves 
 produced by different methods.
 The comparison indicates that IC-Ave outperforms all others  
 for this example. It is interesting to note that 
 although  IC-Last is also based on one-time imputation, it is much better than  
 the median filling, regression tree and BPCA methods. This suggests that for microarray data, 
 the model-based imputation method is potentially more accurate than other one-time imputation methods.  
 The misgLasso algorithm does not work well for this example. This inferiority 
 is not due to the EM algorithm, but due to the gLasso algorithm which 
 does not work well for the example. This is consistent with Liang et al. (2015), where it 
 is shown that the $\psi$-learning algorithm works much better than gLasso for the 
 complete data version of this example.

 \subsection{Yeast Cell Expression Data}

  Gasch et al. (2000) explored genomic expression patterns in the yeast {\it Saccharomyces cerevisiae} 
  responding to diverse environmental changes. 
  % They used DNA microarrays to measure the changes in transcript 
  % levels over time for 6152 known or predicted yeast genes, as cells responded to temperature shocks, 
  % hydrogen peroxide, nitrogen source depletion, etc. 
  % The whole dataset consists of 173 samples collected under 
  % different environmental settings, and is available at 
  The whole dataset has a missing rate of 3.01\% and is available at 
  {\tt http://genome-www.stanford.edu/yeast-stress/}. 
  % In our analysis, we considered only 1000 genes for which the expression levels have the 
  % largest variation over the samples. The missing rate for the selected sub-dataset is also 3.01\%.   
  Our numerical results for a subset of 1000 genes, 
  reported in the Supplementary Material, indicate that the IC algorithm 
  works reasonably well for this example with a few hub genes successfully identified, 
  which are expected to play an important role for yeast cells in response to environmental changes. 
  
% \vspace{-0.25in}
 
\section{High-Dimensional Variable Selection in Presence of Missing Data}

 This problem is also motivated by microarray 
 data analysis, but the goal has been shifted to selection of genes relevant to a
 particular phenotype. 
 To be more general, we let
 $\bY=(Y_1,\ldots,Y_n)^T$ denote the response vector for $n$ observations, and let 
 $\bX=(X_1,\ldots, X_n)^T$ denote the matrix of covariates, where each $X_i$ is 
 a $p$-dimensional vector and $p$
 can be much larger than $n$ (a.k.a. small-$n$-large-$p$).  
 The response variable and covariates are linked through the regression, 
 \begin{equation} \label{Lineareq}
  \bY=(\bm{1}_n, \bX) \bbeta+\bepsilon,
 \end{equation}
  where $\bbeta=(\beta_0,\beta_1,\ldots,\beta_p)^T$ denotes the vector of regression coefficients, 
  and $\bepsilon \sim N(0, \sigma_{\epsilon}^2 I_n)$ denotes the vector of random errors.  
 
 Variable selection for the model (\ref{Lineareq}) with complete data 
 has been extensively studied in the recent literature. Methods have been developed from both frequentist and Bayesian perspectives,
 see e.g., Tibshirani (1996), Johnson and Rossell (2012), and Song and Liang (2015a). 
 For incomplete data, Garcia et al. (2010) proposed to conduct variable selection by 
 maximizing the penalized likelihood function of the incomplete data. However,  
 when $p$ is large and the covariates $X_i$'s are generally correlated, the 
 incomplete data likelihood function can be intractable, rendering failure of 
 their method. Zhao and Long (2013) showed through numerical studies that 
 for the high dimensional data the standard multiple imputation approach performs poorly,
 while the imputation method based on Bayesian Lasso often works better. 
 However, since Bayesian Lasso tends to over-shrink the non-zero regression coefficients, 
 its consistency in variable selection is hard to be justified when $p$ is 
 much greater than $n$ (Castillo et al., 2015). 
 Quite recently, Long and Johnson (2015) proposed to combine 
 Bayesian Lasso imputation and stability selection (Meinshausen and B\"uhlmann, 2010). 
 Again, the consistency of this method is hard to be justified due to the 
 inconsistency of Bayesian Lasso.   

 \subsection{The ICC Algorithm}

 In what follows, we consider a general setting of the model (\ref{Lineareq}), where the covariates 
 follow a multivariate Gaussian distribution  $\bX \sim N(\bmu, \bSigma)$. 
 Under this setting, the parameter vector $\btheta$ consists of 
 three natural blocks  $\bbeta$, $\sigma_{\epsilon}^2$ and the concentration matrix $\bC=\bSigma^{-1}$. 
 Since $n$ has been assumed to be smaller than $p$, 
 we further assume the sparsity for both the regression coefficients $\bbeta$ and 
 the concentration matrix $\bC$. 
 
 To apply the ICC algorithm to this problem, we choose the SIS-MCP algorithm as 
 the consistent estimator of $\bbeta$. That is, the variables are first subject to a sure independence 
 screening procedure, and then the survived variables are selected using the MCP method (Zhang, 2010). 
 This algorithm has been 
 implemented in the R-package {\it SIS}. 
 Given an estimates of $\bbeta$, $\sigma_{\epsilon}^2$ can be  
 estimated by $\hat{\sigma}_{\epsilon}^2=\sum_{i=1}^n \hat{\epsilon}_i^2/(n-|\hat{\bbeta}|-1)$,  
 where $\hat{\epsilon}_i$ denotes the residual of sample $i$, and 
 $|\hat{\bbeta}|$ denotes the number of nonzero elements included in the estimate $\hat{\bbeta}$. 
 Given the consistency of $\hat{\bbeta}$, the consistency of 
 $\hat{\sigma}_{\epsilon}^2$ is easy to be justified.   
 To estimate the concentration matrix $\bC$, we 
 choose the $\psi$-learning algorithm. As mentioned previously, the $\psi$-learning 
 algorithm provides a consistent estimate for the  
 Gaussian graphical network, based on which a consistent estimate of the concentration matrix 
 can be uniquely determined by the algorithm given in Hastie et al. (2009, p.634). 
 Note that SIS-MCP does not make use of the dependency among the covariates. 
 Given the structure of the ICC algorithm, some other 
 variable selection algorithms which have made use of 
 the dependency among the covariates, e.g., Yu and Liu (2016), 
 can also be applied here.  

 Next, we consider the imputation step. 
 Suppose that the value of $x_{hk}$ is missed in $\bX$. Section 4 of the Supplementary Material 
 presents the conditional distributions of $X_{hk}$ given $\bY$ and the rest elements 
 of $\bX$ under different scenarios. 
 Based on the conditional distributions, $x_{hk}$ can be easily imputed by 
 sampling from the respective samplized conditional distributions.  
 Here the samplized conditional distribution refers to the distribution with its 
 population parameters replaced by 
 their respective estimates calculated from samples. For example, $\beta_i$'s 
 are replaced by their SIS-MCP estimates, $\sigma_{\epsilon}^2$ is replaced 
 by $\hat{\sigma}_{\epsilon}^2$, etc. 
 In summary, the ICC algorithm works as follows:
 
 {\it
 \begin{itemize}
 \item  (Initialization) Fill each missing entry of $\bX$ by the median of 
        the corresponding variable, and then iterates between the CC- and I-steps.

 \item (CC-step) (i) Apply the SIS-MCP algorithm to estimate the regression coefficients  $\bbeta$;
       (ii) estimate $\sigma_{\epsilon}^2$ conditional on the estimate of $\bbeta$; 
       and (iii) apply the $\psi$-learning algorithm to learn the structure of the
       Gaussian graphical network.

 \item (I-step) Impute missing values according to the conditional distributions 
       (given in the Supplemental Material) based on the regression model 
        and network structure learned in the CC-step.
 \end{itemize}
 }

\subsection{A Simulated Example}
 
 The datasets were simulated from the model (\ref{Lineareq}) with $n=100$ and 
 $p$=200 and 500. The covariates $\bX$ were generated under two settings:
 (i) the covariates are mutually independent, where $\bx_i \sim N(0,2 I_n)$ for $i=1,\ldots,n$; 
 and (ii) the covariates are generated according to the concentration matrix (\ref{plugin}). 
 For both settings, we set $(\beta_0,\beta_1,\ldots,\beta_5)=(1,1,2,-1.5, -2.5, 5)$ and 
 $\beta_6=\cdots=\beta_p=0$, and random error $\bepsilon \sim N(0, I_n)$.   
 For each pair of $(n,p)$, we simulated 10 datasets independently. For each dataset,
 we considered two missing rates,  
 randomly deleting 5\% and 10\% entries of $\bX$ as missing values.  
 The performance of different methods was measured using three criteria: 
 \[ 
 \mbox{err}_{\bbeta}^2=\|\hat{\bbeta}-\bbeta\|^2, \quad 
 \mbox{fsr}=\frac{|\bs \backslash \bs^*|}{|\bs|}, \quad 
  \mbox{nsr}=\frac{|\bs^* \backslash \bs|}{|\bs^*|},
 \]
 where $\|\cdot\|$ denotes the Euclidean norm, $\hat{\bbeta}$ denotes the estimate of $\bbeta$, 
 $\bs^*$ denotes the set of true covariates, 
 and $\bs$ denotes the set of selected covariates.  

 The ICC algorithm was first applied to this example with the results summarized 
 in Table \ref{RegressionTab1} and \ref{RegressionTab2}. 
 For each dataset, the algorithm was run for 30 iterations. 
 For variable selection, we kept only the variables appeared 5 or more times 
 in the last 10 iterations. For estimation of $\bbeta$, we averaged the estimates of 
 $\bbeta$ obtained in the last 10 iterations. 
 For comparison, we also tried 
 the one-time imputation methods, including median filling and BPCA. As explained previously,
 the median filling method is to fill each missing value by the median of the corresponding variable, 
 and BPCA is to impute the missing values based on the principal component regression.  
 Then the variables are selected using the SIS-MCP method. 
 
\begin{table}[htbp]
\begin{center}
\caption{Comparison of the ICC algorithm with the median filling and BPCA methods for high-dimensional 
 variable selection with independent covariates. ''True'' denotes the results obtained by the 
 MCP method from the complete data.  The values in the table are obtained by averaging over 10
 independent datasets with the standard deviation reported in the parentheses.}
\label{RegressionTab1}
\vspace{2mm}
\begin{tabular}{ccccccc} \hline
 $p$ &Missing Rate & & BPCA & Median & ICC & True \\ \hline
  \multirow{6}{2cm}{\centering $200$} && $\mbox{err}_{\bbeta}^2$ & 0.257(0.267) & 0.262(0.261) & 0.042(0.041) & 0.046(0.048)  \\
  &5\%&fsr & 0.119(0.143)& 0.082(0.092) & 0(0) & 0(0)  \\ 
  && nsr & 0(0) & 0(0) & 0(0) & 0(0) \\  \cline{2-7}
    && $\mbox{err}_{\bbeta}^2$  & 0.903(0.396) & 0.856(0.421) & 0.065(0.087) & 0.046(0.048)  \\ 
  &10\%&fsr & 0.310(0.159) & 0.308(0.178) & 0(0) & 0(0) \\ 
  && nsr & 0(0) & 0(0) & 0(0) & 0(0) \\ \hline
 % \multirow{6}{2cm}{\centering $300$} && $\mbox{err}_{\bbeta}^2$ & 0.297(0.213) & 0.279(0.206) & 0.032(0.022) & 0.022(0.015)  \\ 
 % &5\%&fsr & 0.221(0.221) & 0.252(0.203) &  0(0) & 0(0)  \\ 
 % && nsr & 0(0) & 0(0) & 0(0) & 0(0) \\ \cline{2-7}
 %   && $\mbox{err}_{\bbeta}^2$ & 0.507(0.335) & 0.482(0.267) & 0.055(0.026) & 0.022(0.015)  \\ 
 % &10\%&fsr & 0.260(0.170) & 0.292(0.189) &  0(0) & 0(0)  \\ 
 % && nsr & 0(0) & 0(0) & 0(0) & 0(0) \\ \hline
 % \multirow{6}{2cm}{\centering $400$} && $\mbox{err}_{\bbeta}^2$ & 0.300(0.279) & 0.266(0.276) & 0.024(0.011) & 0.022(0.017)  \\ 
 % &5\%&fsr & 0.311(0.192) & 0.343(0.115) & 0(0) & 0(0) \\ 
 % &&nsr & 0(0) & 0(0) & 0(0) & 0(0)\\ \cline{2-7}
 %   && $\mbox{err}_{\bbeta}^2$ & 0.768(0.500) & 0.847(0.655) & 0.042(0.017) & 0.022(0.017)  \\ 
 % &10\%&fsr & 0.359(0.185) & 0.332(0.240) & 0(0) & 0(0) \\ 
 %  &&nsr & 0.167(0.052) & 0.033(0.070) & 0(0) & 0(0)\\  \hline
  \multirow{6}{2cm}{\centering $500$} && $\mbox{err}_{\bbeta}^2$ & 0.339(0.214) & 0.350(0.206) & 0.029(0.034) & 0.027(0.023)  \\ 
  &5\%&fsr & 0.249(0.225) & 0.266(0.237) & 0(0) & 0(0)   \\ 
  && nsr & 0(0) & 0(0) & 0(0) & 0(0) \\ \cline{2-7}
    && $\mbox{err}_{\bbeta}^2$ & 1.532(1.071) & 1.354(0.895) & 0.044(0.022) & 0.027(0.023)  \\ 
  &10\%&fsr & 0.470(0.265) & 0.420(0.255) & 0(0) & 0(0)  \\ 
  && nsr & 0.033(0.070) & 0.017(0.053) & 0(0) & 0(0) \\ \hline
\end{tabular}
\end{center}
\end{table}

\begin{table}[htbp]
\begin{center}
\caption{Comparison of the ICC algorithm with the median filling and BPCA methods for high-dimensional
 variable selection with dependent covariates. ''True'' denotes the results obtained by the
 MCP method from the complete data.}
\label{RegressionTab2}
\vspace{2mm}
\begin{tabular}{ccccccc} \hline
 $p$ &Missing Rate & & BPCA & Median & ICC & True \\ \hline

  \multirow{6}{2cm}{\centering $200$} && $\mbox{err}_{\bbeta}^2$ & 0.580(0.413) & 0.548(0.140) & 0.118(0.097) & 0.071(0.050)  \\ 
  &5\%&fsr & 0.262(0.204)& 0.263(0.200) & 0(0) & 0(0)  \\ 
  && nsr & 0.017(0.052) & 0.017(0.052) & 0(0) & 0(0) \\ \cline{2-7}
    && $\mbox{err}_{\bbeta}^2$ & 1.604(0.666) & 1.575(0.974) & 0.424(0.461) & 0.071(0.050)  \\ 
  &10\%&fsr & 0.247(0.229) & 0.273(0.238) & 0(0) & 0(0) \\ 
  && nsr & 0.100(0.086) & 0.083(0.088) & 0.033(0.070) & 0(0) \\ \hline
 % \multirow{6}{2cm}{\centering $300$} && $\mbox{err}_{\bbeta}^2$ & 0.971(0.777) & 0.692(0.352) & 0.179(0.125) & 0.114(0.081)  \\ 
 % &5\%&fsr & 0.280(0.259) & 0.254(0.254) & 0(0) & 0.014(0.045)  \\ 
 % && nsr & 0.017(0.053) & 0(0) & 0(0) & 0(0) \\ \cline{2-7}
 %   && $\mbox{err}_{\bbeta}^2$ & 1.338(0.688) & 1.412(0.709) & 0.704(0.569) & 0.114(0.081)  \\ 
 % &10\%&fsr & 0.308(0.247) & 0.324(0.255) & 0(0) & 0.014(0.045)  \\ 
 % && nsr & 0.067(0.086) & 0.067(0.086) & 0.067(0.086) & 0(0) \\ \hline
 % \multirow{6}{2cm}{\centering $400$} && $\mbox{err}_{\bbeta}^2$ & 1.170(0.766) & 1.134(0.726) & 0.394(0.458) & 0.146(0.163)  \\ 
 % &5\%&fsr & 0.135(0.158) & 0.119(0.166) & 0(0) & 0.014(0.045) \\ 
 % && nsr & 0.100(0.086) & 0.100(0.086) & 0(0) & 0(0) \\ \cline{2-7}
 %   && $\mbox{err}_{\bbeta}^2$ & 1.383(1.082) & 1.617(1.225) & 0.417(0.474) & 0.146(0.163)  \\ 
 % &10\%&fsr & 0.362(0.250) & 0.360(0.214) & 0(0) & 0.014(0.045)  \\
 % && nsr & 0.067(0.086) & 0.067(0.086) & 0.033(0.070) & 0(0) \\ \hline
  \multirow{6}{2cm}{\centering $500$} && $\mbox{err}_{\bbeta}^2$ & 0.669(0.366) & 0.717(0.358) & 0.172(0.195) & 0.096(0.083)  \\ 
  &5\%&fsr & 0.262(0.202) & 0.289(0.236) & 0(0) & 0(0)  \\ 
  && nsr & 0.017(0.053) & 0.017(0.053) & 0(0) & 0(0)\\ \cline{2-7}
    && $\mbox{err}_{\bbeta}^2$ & 2.752(2.306) & 2.896(2.601) & 0.578(0.587) & 0.096(0.083)  \\ 
  &10\%&fsr & 0.297(0.230) & 0.327(0.224) & 0(0) & 0(0)  \\ 
  && nsr & 0.133(0.070) & 0.133(0.070) & 0.050(0.081) & 0(0) \\ \hline
\end{tabular}
\end{center}
\end{table}

The comparison indicates that the ICC algorithm works extremely well for this example. 
For the case of independent covariates, its results are almost as good as those 
 obtained from the complete data. In both cases, the ICC algorithm 
 significantly outperforms the one-time imputation methods. 

\subsection{A Real Data Example} 

We analyzed one real gene expression dataset about Bardet-Biedl syndrome (Scheetz et al., 2006).  
The complete dataset contains 120 samples, where the expression level of the gene TRIM32 
works as the response variable and the expression levels of 200 other genes work as the predictors. 
The dataset is available in the R package {\it flare}. 
We generated ten incomplete datasets from the complete one by 
randomly deleting 5\% observations. For each incomplete dataset, we ran the ICC algorithm 
for 30 iterations and averaged the estimates of $\bbeta$ obtained in the last 
10 iterations as the final estimate. For comparison, the median filling and BPCA methods 
 were also applied to this example. Table \ref{eyetab} summarizes the estimation errors 
 of $\hat{\bbeta}$ (with respect to $\bbeta_c$, the estimate of $\bbeta$ from the complete data) 
 produced by the three methods for ten incomplete datasets.     

 \begin{table}[htbp]
\begin{center}
\caption{Estimation errors of $\hat{\bbeta}$ (with respect to $\bbeta_c$)
 produced by ICC, median filling and BPCA for the Bardet-Biedl syndrome example, where 
 err$_{\bbeta}^2$ is calculated by averaging $\|\hat{\bbeta}-\bbeta_c\|^2$ over ten incomplete datasets, 
 and ``s.d.'' represents the standard deviation of err$_{\bbeta}^2$. }
\label{eyetab}
\vspace{2mm}
\begin{tabular}{cccc} \hline
  Method  & BPCA & Median & ICC \\ \hline
  err$_{\bbeta}^2$  &  0.428   & 0.397  &  0.187   \\
  s.d.              &  0.091   & 0.086  &  0.040  \\ \hline
\end{tabular} 
\end{center}
\end{table}

 We have also explored the results of variable selection.  
 The complete data model selects 5 variables: v.153, v.180, v.185, v.87 and v.200.
 For the ICC, median filling and BPCA models, we count the selection frequency of each variable 
 for the ten incomplete datasets. 
 For the ICC models, the top 5 variables in selection frequency 
 are v.153, v.185, v.180, v.87 and v.200, which are the same (ignoring the order) as the complete data model.  
 For the median filling models, the top 5 variables are v.153, v.185, v.62, v.200 and v.54. 
 For the BPCA models, the top 5 variables are v.153, v.87, v.185, v.62 and v.200. 
 Both the results of $\bbeta$ estimation and variable selection indicate the 
 superiority of the ICC algorithm over the one-time imputation methods.

\section{A Random Coefficient Linear Model} 
 
 To further illustrate the use of the ICC algorithm, we consider  
 a random coefficient linear model. Such a model often arises, for instance, in recommendation 
 systems where the customers rate different items, e.g., products or service.  Specifically, we simulate 
 the data from the following model 
 \begin{equation} \label{randomcoefmodel} 
 \begin{split}
 &  y_{ij}=\bx_{ij}^T \bbeta+\bz_i^T \blambda_i +\bw_j^T \bgamma_j +e_{ij}, \\
 & e_{ij}\sim N(0, \sigma^2), \quad \blambda_i \sim N(0,\Lambda), \quad \bgamma_j \sim N(0,\Gamma), \\
 \end{split}
 \end{equation}
 where $y_{ij}$ represents the response for customer $i$ on item $j$. Assuming that there 
 are $I$ customers and each customer responds to $J$ items. Thus,  
  the dataset consists of a total of $n=IJ$ observations.  
 The vector $\bx_{ij}$ represents the covariates that characterize the customers and items, e.g., 
 how and how long the customer has purchased the item;  
 $\bz_i$ represents customer-specific covariates such as gender, education and demographics; and 
 $\bw_j$ represents item-specific covariates, e.g., the manufacturer and category of the item.   
 The vector $\blambda_i$ represents the customer-specific (random) coefficients 
 and $\bgamma_j$ represents the item-specific (random) coefficients.  
 This model can be easily extended to the case where each customer responds to only 
 a subset of items. For this model, we treat  
 the random coefficients $\blambda_i$'s and $\bgamma_j$'s as missing data, 
 and are interested in estimation of $\bbeta$. 
 For simplicity, we assume that $\bbeta$ is low-dimensional, 
 although the whole dataset can be big when $I$ and/or $J$ become large. 
 Under this assumption, the ICC algorithm 
 is essentially reduced to the stochastic EM algorithm for this example. 
 Instead of using the ICC algorithm in this straightforward way, we propose 
 to use it under the Bayesian framework. This extends the applications of the ICC algorithm 
 to Bayesian computation. 
 
%  Following Ansari et al. (2014), 
 To conduct Bayesian analysis for the model, we assume the following semiconjugate priors:
 \begin{equation} \label{prioreq}
 \bbeta \sim N(\bmu_{\bbeta}, \Sigma_{\bbeta}), \quad \sigma^2 \sim IG(a,b), \quad 
 \Lambda \sim IW(\rho_{\Lambda}, \bR_{\Lambda}), \quad \Gamma \sim IW(\rho_{\Gamma}, \bR_{\Gamma}),
 \end{equation}
 where $IG(\cdot,\cdot)$ denotes the inverted Gamma distribution, 
 $IW(\cdot,\cdot)$ denotes the inverted Wishart distribution, and $\bmu_{\bbeta}$, 
 $\Sigma_{\bbeta}$, $a$, $b$, $\rho_{\Lambda}$, $\bR_{\Lambda}$, $\rho_{\Gamma}$, 
 and $\bR_{\Gamma}$ are hyperparameters to be specified by the user.  
 Each of these priors is individually conjugate to the normal likelihood function, 
 given the other parameters, although the joint prior is not conjugate. 
 Given these priors, the full conditional posterior distributions are derived in 
 Section 5 of the Supplementary Material. Since, under the low-dimensional setting, 
 the mode of the full conditional posterior distribution provides 
 a consistent estimator for the corresponding parameter, the ICC algorithm can work as 
 follows: 
 
 \begin{itemize}
 \item {\it (Initialization) Initialize $\lambda_i$'s, $\gamma_j$'s, and all parameters by 
       some random numbers. }
 
 \item {\it (CC-step) Estimate the parameters $\bbeta$,
 $\Lambda$, $\Gamma$, and $\sigma^2$ by the mode of their respective full conditional posterior 
 distributions. }

 \item {\it (I-step) Impute the values of $\blambda_i$'s and $\bgamma_j$'s according to their 
  respective full conditional posterior distributions. }
 \end{itemize}

  Under the Bayesian framework, the ICC algorithm works in a similar way to the Gibbs sampler 
  except that it replaces posterior samples of the parameters by their respective  
  full conditional posterior modes. Also, in this case, it is reduced to 
  a hybrid of data augmentation (Tanner \& Wong, 1987) and
  iterative conditional modes (Besag, 1974) by using imputation for the
  missing data and conditional modes for the parameters.    
  However, the ICC algorithm offers more, whose consistency step 
  allows it to conduct parameter estimation based on sub-samples only and this 
  can create great savings in computation for big data problems. 
  For high-dimensional problems, if the choice of prior distributions ensures 
  posterior consistency, then the above algorithm can still be employed.

%\begin{figure}[htbp]
%\begin{center}
%\begin{tabular}{c}
%\epsfig{figure=mixedpath.ps,height=6.0in,width=2.0in,angle=270} 
%\end{tabular}
%\caption{Comparison of the ICC algorithm and the Gibbs sampler for the simulated random coefficient 
% linear model example: (a) \& (b): Sampling paths of $\beta_0$ and $\beta_1$, 
%  where the solid line is for the Gibbs sampler
%  and the dotted line is for the ICC algorithm; (c)\&(d): autocorrelation plots of $\beta_0$ samples generated 
%  by the ICC and Gibbs sampler. The samples used in the plots have been thinned by a factor of 50. }
%\label{randomcoef} %% label for entire figure
%\end{center}
%\end{figure}

  Figure 3 of the supplementary material compares the
  sampling path and autocorrelations of the ICC and Gibbs samples.
  As expected, the comparison shows that the ICC algorithm can converge faster than the Gibbs sampler and, 
  in addition, the samples generated by the ICC algorithm
  tend to have smaller variations than those by the Gibbs sampler.
 We are aware that the accuracy of the ICC estimates is achieved at the price that we scarify the variance
 information contained in the posterior samples. 
 As pointed out by Nielsen (2000), the variance of the ICC samples reflects 
 the information loss due to the missing data. 
 However, for the random coefficient example, the variance information can be obtained from the full conditional posterior 
 distributions (given in the Supplementary Material) by simply plugging the parameter estimates into 
 their variances. 
 % For example, the posterior variance of $\bbeta$ can be obtained by plugging 
 % the estimate of $\sigma^2$ into $\tilde{\Sigma}_{\bbeta}$ given in equation (\ref{vareq1}). 
 This observation suggests a general strategy to 
 improve simulations of the Gibbs sampler: At each iteration, we only need to draw samples for the components 
 for which the posterior variance is not analytically available and also of interest to us, and   
 the other components can be replaced by the mode of the respective full conditional posterior 
 distributions. As aforementioned, the mode can be found with a subset of samples, which can be much cheaper than 
 sampling from the full data conditional posterior.  We expect that the proposed strategy can significantly   
 facilitate Bayesian computation for big data problems.  
 A further study of this proposed strategy will be reported elsewhere.

\section{Discussion}

 In this paper, we have proposed the imputation-consistency algorithm, or the IC algorithm in short,
 as a general algorithm for dealing with high-dimensional missing data problems.  
 Under quite general conditions, we show that the IC algorithm can lead to a consistent estimate 
 for the parameters. We have also extended the IC algorithm to the case of multiple block 
 parameters, which leads to the imputation-conditional consistency (ICC) algorithm. 
 We illustrate the proposed algorithms using the 
 high-dimensional Gaussian graphical models, high-dimensional variable selection,
 and a random coefficient model.  

 Like the EM algorithm for low-dimensional data, we expect that the IC/ICC algorithm can
 have many applications for high-dimensional data. With the IC/ICC algorithm,
 many problems can be much simplified, e.g., variable selection for high-dimensional mixture regression
 (Khalili and Chen, 2007) and variable selection for high-dimensional mixed
 effect models (Fan and Li, 2012). For the former, the group index of each
 sample can be treated as missing data, and then the IC algorithm can be applied:
 The I-step is to assign the samples into different groups and the C-step is to conduct variable selection for
 each group separately. For the latter, at each iteration of the IC algorithm,
 it is reduced to variable selection for a high-dimensional regression
 with fixed designs given imputed random effects.

  To assess the convergence of IC simulations, 
  we recommend the Gelman-Rubin statistic (Gelman and Rubin, 1992).
  Since the IC algorithm usually converges very fast, we do not recommend a long run.
  Our experience shows that 20 iterations have often been long enough to produce a stable 
  parameter estimate. Due to the MCMC nature of IC simulations, we recommend the averaging method 
  for parameter estimation, and allow a burn-in period before collecting 
  the samples for averaging. 
  In Section 3.2, we reported the results based on the last iteration is just to compare
  with other one-time imputation methods.  
  Theoretically, the variation of the IC estimates collected at different iterations  
  reflects the missing data information. 
  %When a large amount of noises
  %were brought into the system through the missing data, the IC estimates
  %will have a large variation along with iterations. 
  Hence, in addition to the parameter estimates, one might report their variance over iterations. However,
  this information is often not of interest to us, we choose not to report them in the paper. 

  Please be aware that when a large amount of noise was brought into 
  the system through missing data, the IC/ICC algorithm may fail to work 
  as implied by condition (A3). Also, please be aware that 
  the IC/ICC algorithm targets consistency by design. When the sample size 
  is small, the consistent estimators might not be adequately accurate. 
  Our numerical experience shows that in this case, 
  the IC/ICC algorithm might not significantly outperform one-time imputation 
  methods, such as median filling and BPCA, 
  as there is no much information to use for improving imputation during iterations.
  In general, when the sample size increases, the IC/ICC algorithm can 
  significantly outperform one-time imputation methods. 

  Regarding statistical inference for parameters, we note that, theoretically, it can be 
  done in general scenarios, not limited to that the posterior distribution 
  of the parameters has a closed form.
  For example, for high-dimensional regression,                                     
  if the Lasso algorithm is employed as the consistent estimator in the IC algorithm, then the 
  inference for $\btheta$, e.g., constructing                     
  confidence intervals for each component of $\btheta$, can be done using the 
  de-sparsified method (Zhang and Zhang, 2014; van de Geer et al., 2014). 
  % based on the pseudo-complete data and the Lasso estimate $\hat{\btheta}$. 
  Let $V(\bxobs, \tbxmis_t,\btheta_n^{(t)})$ denote an uncertainty assessment statistic
  obtained at iteration $t$. 
  Assume that $V(\bxobs, \bxmis, \btheta)$ is a Lipschitz function with respect to 
  $\btheta$ and it is integrable. Then, by Corollary \ref{corMCMC1} or 
  Corollary \ref{corMCMC2} (depending on IC or ICC being used), we will be 
  able to get  an uncertainty assessment for $\btheta$ by 
  averaging $V(\bxobs, \tbxmis_t,\btheta_n^{(t)})$ along the IC/ICC chain. 
  However, how to get the uncertainty assessment statistic 
  $V(\bxobs, \tbxmis_t,\btheta_n^{(t)})$ for general 
  high-dimensional problems is beyond the scope of this paper. 

 As a variant of the ICC algorithm, we note that the I-step can be replaced by an E-step
 if it is available. In this case, the C-step is to maximize the objective function
 \begin{equation} \label{nsfeq1}
 \max_{\btheta} E_{\btheta^*} Q(\bxobs, \btheta|\btheta_n^{(t)}),
 \end{equation}
 where the Q-function is as defined in the EM algorithm,
 and $E_{\btheta^*}$ denotes
 expectation with respect to the true distribution $\pi(\bxobs|\btheta^*)$.        
 Suppose that $\btheta$ has been partitioned into a few blocks $\btheta=(\theta^{(1)}, \ldots, \theta^{(k)})$.
 To solve (\ref{nsfeq1}), the ICC algorithm is reduced to a blockwise consistency algorithm,
 which is to iteratively find
 consistent estimates for the parameters of each block conditioned on the current
 estimates of the parameters of other blocks. 
 The blockwise consistency algorithm is closely related to, but more flexible than, 
 the coordinate ascent algorithm (Tseng, 2001; Tseng and Yun, 2009).   
 The coordinate ascent algorithm is to iteratively find 
 the exact maximizer for each block conditioned on the current estimates of the parameters 
 of other blocks. Under appropriate conditions, such as contraction as in (A5$'$) and 
 uniform consistency of $\btheta_n^{(t)}$'s as shown in Theorem \ref{them1new}, 
 we will be able to show that the paths of the two algorithms 
 will converge to the same point. This will be explored elsewhere. 

 The ICC algorithm have strong implications for big data computing. 
 Based on the ICC algorithm, we have proposed a general strategy to improve Bayesian computation 
 under big data scenario; that is, we can replace posterior samples by posterior modes
 in Gibbs iterations to accelerate simulations, where the posterior modes can be 
 calculated with a subset of samples. 
 In addition, the IC/ICC algorithm facilitates data integration 
 from multiple sources when missing data are present. 
 With the IC/ICC algorithm, the problem of data integration for incomplete data is 
 converted to a problem for complete data and thus 
 many of the existing meta-analysis methods can be conveniently applied for inference.  
 This is very important for big data analysis.

 \section*{Acknowledgement}
 Liang's research was support in part by the grants
 DMS-1545202, DMS-1612924 and DMS/NIGMS R01-GM117597. 
 The authors thank the editor, associate editor and three referees for their  constructive 
 comments which have led to significant improvement of this paper. 

\appendix

\section*{Appendix}

\paragraph{1. Proof of Consistency of $\btheta_n^{(t+1)}$.} 
 
 Define $\Theta_n$ as the parameter space of $\btheta$, where the subscript $n$ indicates the dependence of 
 the dimension of $\btheta$ on the sample size $n$.
 Without possible confusion, we will refer to $\Theta_n$ as $\Theta$. In addition, we let 
 $\Theta_n^T=\{ \theta_n^{(1)}, \ldots, \theta_n^{(T)}\}$ denote a path of $\theta_n$ in the IC algorithm,
 which can be considered as an arbitrary subset of $\Theta_n$ with $T$ elements (replicates are allowed). 
 Let $\tilde{x}=(\xobs,\txmis)$ and define 
 \begin{equation} \label{Liangeq1}
 \begin{split} 
 G_n(\btheta|\btheta_n^{(t)}) & =E_{\btheta_n^{(t)}}\log f_{\btheta}(\tilde{x}) =\int \log (f_{\btheta}(\tilde{x}) ) 
 f(\xobs|\btheta^*) h(\txmis|\btheta_n^{(t)}) d \tilde{x}, \\
 %%%%%%%
 \hat{G}_n(\btheta|\tbx,\btheta_n^{(t)}) &
 =\frac{1}{n}\sum_{i=1}^n  \log f(\xobs_i,\txmis_i|\btheta), \\ 
 %%%%%%%
 \tilde{G}_n(\btheta|\btheta_n^{(t)}) &=\frac{1}{n}\sum_{i=1}^n 
  \int\log f(\xobs_i,\txmis|\btheta)h(\txmis|\xobs_i,\btheta_n^{(t)})d\txmis :=
  \frac{1}{n} \sum_{i=1}^n q(\xobs_i).  
 \end{split}
\end{equation} 
 Our first goal is to show the following uniform law of large numbers (ULLN) holds for any 
 $T$ such that $\log T =o(n)$:    
 \begin{equation} \label{Liangeq2}
 \sup_{\btheta_n^{(t)}\in \Theta_n^T}\sup_{\btheta\in \Theta_n}|\hat G_n(\btheta|\tbx,\btheta_n^{(t)})
 -G_n(\btheta|\btheta_n^{(t)}) |\rightarrow_p 0,  
 \end{equation}  
 where $\rightarrow_p$ denotes convergence in probability. 
 To achieve this goal, we need the following conditions: 
 
 \begin{itemize} 
 \item[(A1)] $\log f_{\btheta}(\tilde{x})$ is a continuous function of $\btheta$ for each $\tilde{x}\in \mathcal{X}$
 and a measurable function of $\tilde{x}$ for each $\theta$.

\item[(A2)] [Conditions for Glivenko-Cantelli theorem]
 \begin{itemize}
 \item[(a)] There exists a function $m_n(\tilde{x})$ such that
  $\sup_{\btheta\in \Theta_n,\tilde{x}\in \mathcal{X}}|\log f_{\btheta}(\tilde{x})|\leq m_n(\tilde{x})$.

 \item[(b)]  Define $\tilde{m}_n(\xobs,\btheta_n^{(t)})=\int m_n(\tilde{x})$ $h(\txmis|\xobs,\btheta_n^{(t)})d \txmis$.
 Assume that there exists $m^*_n(\xobs)$ such that $0\leq \tilde{m}_n(\xobs,\btheta_n^{(t)}) 
 \leq m^*_n(\xobs)$ for all $\btheta_n^{(t)}$, $E[m^*_n(\xobs)]$ $<\infty$,
 and $\sup_{n\in \mathbb{Z}^+}E[m_n^*(\xobs)$  $1(m_n^*(\xobs)\geq \zeta)]\rightarrow 0$
 as $\zeta\rightarrow \infty$.
In addition,
$\sup_{n\geq 1}\sup_{x\in \mathcal{X},\btheta\in \Theta_n}$
 $ |\int m_n(\tilde{x})1(m_n(\tilde{x})>\zeta)h(\txmis|x,\btheta)
 d\txmis|\rightarrow 0$ as $\zeta\rightarrow \infty$.

\item[(c)] Define $\mathcal{F}_n=\{\int \log f(\xobs,\txmis|\btheta)h(\txmis|\xobs$, $\btheta_n^{(t)})d\txmis|\btheta,
 \btheta_n^{(t)}\in \Theta_n\}$, and
 $\mathcal{G}_{n,M}=\{q1(m^*_n(\xobs)$ $\leq M)|q\in \mathcal{F}_n\}$. Suppose that for every $\epsilon$ and $M>0$, 
 the metric entropy $\log N(\epsilon,\mathcal{G}_{n,M},L_1(\mathbb{P}_n))=o_p^*(n)$, 
 where $\mathbb{P}_n$ is the empirical measure of $\xobs$,
% $L_1(\cdot)$ denotes the $L_1$-norm, 
 and $N(\epsilon,\mathcal{G}_{n,M},L_1(\mathbb{P}_n))$ is the covering number with respect to the 
 $L_1(\mathbb{P})$-norm. 
 % i.e., the minimum number of balls
 % $\{g: \|g-q\| \leq \epsilon \}$ of radius $\epsilon$ needed to
 % cover the set $\mathcal{G}_{n,M}$.
 \end{itemize}

\item[(A3)] [Conditions for imputed data]  
 Define $Z_{t,i}=\log f(\xobs_i,\txmis_i|\btheta)-\int \log f(\xobs_i,\txmis|\btheta) 
  h(\txmis| $  $\xobs_i, \btheta_n^{(t)})$ $d\txmis$. Suppose that  for any $\btheta$ and 
  $\btheta_n^{(t)} \in \Theta_n$,  
  $E|Z_{t,i}|^m \leq m! M_b^{m-2} v_i/2$ for every $m\geq 2$ and some constants $M_b>0$ and $v_i=O(1)$.  
   That is, $Z_{t,i}$'s are sub-exponential random variables. 
 \end{itemize}

 \begin{theorem} \label{them0}
 Assume conditions (A1)--(A3), then (\ref{Liangeq2}) holds for 
 any $T$ such that $\log T =o(n)$. 
 \end{theorem} 
 \begin{proof} 
 By the definitions in (\ref{Liangeq1}), we have the decomposition 
 \begin{equation} \label{decompeq}
 \hat G_n(\btheta|\tbx,\btheta_n^{(t)})-G_n(\btheta|\btheta_n^{(t)})
 =\big\{\hat{G}_n(\btheta|\tbx,\btheta_n^{(t)})-\tilde{G}_n(\btheta|\btheta_n^{(t)})\big\}
 +\big\{\tilde{G}_n(\btheta|\btheta_n^{(t)})-G_n(\btheta|\btheta_n^{(t)})\big\}, 
 \end{equation}
 which consists of two terms, 
 the first term comes from imputation of missing data,
 and the second term comes from the observed data.

 First, we show that the second term of (\ref{decompeq}) converges to 0 uniformly, following  
 the proof of Theorem 2.4.3 of van der Vaart and Wellner (1996).
 By the symmetrization Lemma 2.3.1 of van der Vaart and Wellner (1996), measurability of the class 
 $\mathcal{F}_n$, and Fubini's theorem, 
 \[
 \begin{split}
 & E^*\sup_{\btheta,\btheta_n^{(t)}\in \Theta_n}|\tilde{G}_n(\btheta|\btheta_n^{(t)})-G_n(\btheta|\btheta_n^{(t)})|
\leq 2 E_{\xobs}E_\epsilon \sup_{q(x)\in \mathcal{F}_n}\|\frac{1}{n}\sum_{i=1}^n \epsilon_i q(x_i^{obs})\| \\
 & \leq 2 E_{\xobs}E_\epsilon \sup_{q(x)\in \mathcal{G}_{n,M}}\|\frac{1}{n}\sum_{i=1}^n \epsilon_i q(x_i^{obs})\|
 +2 E^*[m^*_n(x^{obs})1(m^*_n(x^{obs})>M)], \\
\end{split} 
 \]
where $\epsilon_i$ are i.i.d. Rademacher random variables 
with $P(\epsilon_i=+1)=P(\epsilon_i=-1)=1/2$, and $E^*$ denotes the outer expectation.

By condition (A2), $2 E^*[m^*_n(\xobs)1(m^*_n(\xobs)>M)]\rightarrow 0$ for sufficiently large $M$.  
To prove convergence in mean, it suffices to show that the first term converges to zero for fixed $M$. 
Fix $\xobs_1,...,\xobs_n$, and let $\mathcal{H}$ be a $\epsilon$-net in $L_1(\mathbb{P}_n)$ 
 over $\mathcal{G}_{n,M}$, then
\[
E_\epsilon \sup_{q(x)\in \mathcal{G}_{n,M}}\|\frac{1}{n}\sum_{i=1}^n \epsilon_i q(\xobs_i)\|\leq E_\epsilon 
 \sup_{q(x)\in \mathcal{H}}\|\frac{1}{n}\sum_{i=1}^n \epsilon_i q(\xobs_i)\|+\epsilon.
\]
The cardinality of $\mathcal{H}$ can be chosen equal to $N(\epsilon,\mathcal{G}_{n,M},L_1(\mathbb{P}_n)$. 
Bound the $L_1$-norm on the right by the Orlicz-norm $\psi_2$ and use the 
maximal inequality (Lemma 2.2.2 of van der Vaart and Wellner (1996)) and Hoeffding's inequality, 
it can be shown that 
\begin{equation} \label{Liangeq5}
 E_\epsilon \sup_{q(x)\in \mathcal{G}_{n,M}}\|\frac{1}{n}\sum_{i=1}^n \epsilon_i q(\xobs_i)\| 
 \leq K  \sqrt{1+\log N(\epsilon,\mathcal{G}_{n,M},L_1(\mathbb{P}_n))}\sqrt{6/n}M +\epsilon  
  \rightarrow_{P^*}  \epsilon,
\end{equation} 
where $K$ is a constant, and $P^*$ denotes outer probability. 
It has been shown that the left side of (\ref{Liangeq5}) converges to zero in probability. Since it is 
 bounded by $M$, its expectation with respect to $\xobs_1,\ldots,\xobs_n$ converges to zero by the 
 dominated convergence theorem.
 
 This concludes the proof that 
$\sup_{\btheta_n^{(t)}\in \Theta_n}\sup_{\btheta\in \Theta_n}|\tilde{G}_n(\btheta|\btheta_n^{(t)})
 -G_n(\btheta|\btheta_n^{(t)})|\rightarrow_p 0$ in mean. Further, by Markov inequality, we conclude that 
 \begin{equation} \label{Liangeq3}
 \sup_{\btheta_n^{(t)}\in \Theta_n}\sup_{\btheta\in \Theta_n}|\tilde{G}_n(\btheta|\btheta_n^{(t)})
 -G_n(\btheta|\btheta_n^{(t)})|  \to_p 0. 
 \end{equation} 
 
 To establish the uniform convergence of the first term of (\ref{decompeq}), we fix 
 $\xobs_1,...,\xobs_n$. By  
  condition (A3), 
 $n(\hat{G}_n(\btheta|\tilde{x},\btheta_n^{(t)})-\tilde{G}_n(\btheta|\btheta_n^{(t)}))= 
 Z_{t,1}+Z_{t,2}+\cdots+Z_{t,n}$. By Bernstein's inequality,   
 \[
  P(n|\hat{G}_n(\btheta|\tilde{x},\btheta_n^{(t)})-\tilde{G}_n(\btheta|\btheta_n^{(t)})|>z)
  =P(|Z_{t,1}+Z_{t,2}+\cdots+Z_{t,n}|>z)\leq 2 \exp\left\{-\frac{1}{2} \frac{z^2}{v+M_b z}\right\},
 \]
 for $v \geq v_1+\cdots +v_n$. 
 By Lemma 2.2.10 of van der Vaart and Wellner (1996), for Orlicz norm $\psi_1$, we have
 \[
 \begin{split}
 & \|\sup_{\btheta\in \Theta_n, t=1,2,\ldots,T} 
  n \big\{\hat{G}_n(\btheta|\tbx,\btheta_n^{(t)})-\tilde{G}_n(\btheta|\btheta_n^{(t)})  \big\}  \|_{\psi_1}  \\
 & \leq \epsilon +
 K(M_b \log (1+TN(\epsilon,\mathcal{G}_{n,M},L_1(\mathbb{P}_n)))
 +\sqrt{v}\sqrt{\log(1+TN(\epsilon,\mathcal{G}_{n,M},L_1(\mathbb{P}_n)))}),
 \end{split}
\]
 for a constant $K$ and any $\epsilon>0$. By condition (A2)-(c) and the 
 condition $\log(T)=o(n)$,  
 \[
 \begin{split}
 & \|\sup_{\btheta\in \Theta_n, t=1,2,\ldots,T} 
  \{ \hat{G}_n(\btheta|\tbx,\btheta_n^{(t)})-\tilde{G}_n(\btheta|\btheta_n^{(t)}) \}  \|_{\psi_1} \\
 & \leq  \epsilon+
  K(M_b \log (1+TN(\epsilon,\mathcal{G}_{n,M},L_1(\mathbb{P}_n)))/n 
 +\sqrt{v/n}\sqrt{\log(1+TN(\epsilon,\mathcal{G}_{n,M},L_1(\mathbb{P}_n)))/n}) 
  \rightarrow_{P^*} \epsilon. 
 \end{split}
\]
 Therefore, 
\begin{equation} \label{Liangeq4}
 \sup_{\btheta\in \Theta_n,t=1,2,...,T}|\hat G_n(\btheta|\widetilde{x},\btheta_{n}^{(t)})
 -\tilde{G}_n(\btheta|\btheta_{n}^{(t)})|\rightarrow_{p} 0.
\end{equation} 
 The theorem can then be concluded by combining (\ref{Liangeq3}) and (\ref{Liangeq4}). 
 \end{proof}

 \noindent 
 {\bf Remark R1} {\it (On the metric entropy condition)}  
  Assume that all elements in $\cup_{n\geq 1}\mathcal{F}_n$ are uniformly Lipschitz with respect 
  to the $L_1$-norm.
  Then the metric entropy $\log N(\epsilon,\mathcal{G}_{n,M},L_1(\mathbb{P}_n))$ can be 
  measured based on the parameter space $\Theta$.  
  Since the functions in $\mathcal{G}_{n,M}$ are all bounded, 
  the corresponding parameter space $\Theta_{n,M}$ can be contained in a $L_1$-ball by 
  the continuity of $\log f(\tilde{x}|\btheta)$ in $\btheta$. 
  Further, we assume that the diameter of the $L_1$-ball or the space $\Theta_{n,M}$ 
  grows at a rate of $O(n^{\alpha})$ for some $0 \leq \alpha<1/2$,
  then $\log N(\epsilon,\mathcal{G}_{n,M},L_1(\mathbb{P}_n)) = O(n^{2\alpha} \log p)$ holds, 
  which allows $p$ to grow at a polynomial rate of $O(n^{\gamma})$ for some constant  $0<\gamma<\infty$.  
  Note that the increased diameter accounts for the conventional assumption 
  that the size of the true model grows with the sample size $n$.  
  Refer to Vershynin (2015) for more discussions on this issue. 
  Similar conditions on metric entropy have been used in the literature 
  of high-dimensional statistics. For example,  
  Raskutti et al. (2011) studied minimax rates of estimation for 
  high-dimensional linear regression over $L_q$-balls.

\noindent {\bf Remark R2} {\it (On the condition of $T$)}.  Since the imputation step draws random data at each iteration $t$,
  there is no way to show uniform convergence of $\btheta_n^{(t+1)}$
  to $\btheta_*^{(t)}$ over all possible $\btheta_n^{(t)} \in \Theta_n$. However, we are able to prove
that the consistency results hold for any sequence of $\btheta_n^{(1)},...,\btheta_n^{(T)}$
with $T$ being not too large compared to $e^n$. This is enough for Theorems \ref{them1}-\ref{them4}.
To justify this, we may consider the case that the dimension of $\btheta_n$ grows with $n$ at a 
rate of $p=O(n^{\gamma})$ for a constant $\gamma>0$, say, $\gamma=5$. 
 Then it is easy to see that when $n>13$,
 the ratio $T/p$ has an order of
\[
 O(e^n/p)=O(e^{n-\gamma \log(n)}) \succ O(e^{0.1 n}) \succ O(p^{100}),
\]
 which implies that essentially there is no constraint on the setting of $T$.
 Note that for MCMC simulations, the number of iterations is often set to a low-order polynomial
 of $p$ for a given set of observations.

 For any $\btheta_n^{(t)}\in \Theta_n^{T}$, 
 we define ${\btheta}_n^{(t+1)}=\arg\max_{\btheta\in \Theta_n} 
 \hat{G}_n(\btheta|\widetilde{x},\btheta_n^{(t)})$
 and $\btheta_*^{(t)}=\arg\max_{\btheta\in \Theta_n} G_n(\btheta|\btheta_n^{(t)})$.
 We would like to establish the uniform consistency of ${\btheta}_n^{(t+1)}$ 
 with respect to $t$, i.e.,
 \begin{equation} \label{Faeq1} 
  \sup_{t \in\{1,2,\ldots,T\}} \|{\btheta}_n^{(t+1)}-\btheta_*^{(t)}\| \to_p 0, 
  \quad \mbox{as $n\to \infty$}.  
 \end{equation}  
 
 To achieve this goal, we assume the following condition: 
 \begin{itemize}
 \item[(A4)] For each $t=1,2,\ldots, T$, $G_n(\btheta|\btheta_n^{(t)})$ 
             has a unique maximum at $\btheta_*^{(t)}$; for any $\epsilon>0$,
           $\sup_{\btheta \in \Theta_n\setminus B_t(\epsilon)} G_n(\btheta|\btheta_n^{(t)})$            exists, where 
           $B_t(\epsilon)=\{\btheta \in \Theta_n: \|\btheta-\btheta_*^{(t)}\| < \epsilon\}$. 
           Let $\delta_t=G_n(\btheta_*^{(t)}|\btheta_n^{(t)})- 
          \sup_{\btheta \in \Theta_n\setminus B_t(\epsilon)} G_n(\btheta|\btheta_n^{(t)})$, 
           $\delta=\min_{t \in \{1,2,\ldots,T\}} \delta_t>0$. 
 \end{itemize}  

 Note that the existence of  $\sup_{\btheta \in \Theta_n\setminus B_t(\epsilon)} 
 G_n(\btheta|\btheta_n^{(t)})$ can be easily satisfied if $\Theta_n$ is restricted 
 to a compact set, which implies that 
 $\Theta_n\setminus B_t(\epsilon)$ is also a compact set and thus the supremum is achievable. 
 This condition can also be 
 satisfied by assuming that $\Theta_n$ is convex and for each $t$,
  $\btheta_*^{(t)}$ is in the 
 interior of $\Theta_n$ and $G_n(\btheta|\btheta_n^{(t)})$ is concave in $\btheta$.  
  
 \begin{theorem} \label{them1new} Assume conditions (A1)-(A4) hold, then the maximum pseudo-complete data 
 likelihood estimate ${\btheta}_n^{(t+1)}$ is uniformly consistent to $\btheta_*^{(t)}$ 
 over $t=1,2,\ldots,T$,  i.e. (\ref{Faeq1}) holds. 
 \end{theorem}  
  
 \begin{proof} Since both $\hat{G}_n(\btheta|\widetilde{x},\btheta_n^{(t)})$ and 
  $G_n(\btheta|\btheta_n^{(t)})$ are continuous in $\theta$ as implied by the continuity of 
 $\log f_{\btheta}(\tbx)$, the remaining part of the proof follows from 
 Lemma \ref{lem00} by setting the penalty function $P_{\lambda_n}(\btheta)=0$ for all 
 $\btheta \in \Theta_n$. \end{proof} 
 
 \begin{lemma}\label{lem00} Consider a sequence of 
 functions $Q_t(\btheta, \bX_n)$ for $t=1,2,\ldots,T$. Suppose that the following 
  conditions are satisfied: (B1) For each $t$, 
  $Q_t(\btheta,\bX_n)$ is continuous in $\btheta$ and there exists a function 
  $Q_t^*(\btheta)$, which is continuous in $\btheta$ and 
  uniquely maximized at $\btheta_*^{(t)}$.
  (B2) For any $\epsilon>0$,  $\sup_{\btheta \in \Theta_n\setminus B_t(\epsilon)} 
  Q_t^*(\btheta)$ exists, where 
   $B_t(\epsilon)=\{\btheta: \|\btheta-\btheta_*^{(t)}\| < \epsilon\}$; 
   Let $\delta_t=Q_t^*(\btheta_*^{(t)})-
    \sup_{\btheta \in \Theta_n\setminus B_t(\epsilon)} Q_t^*(\btheta)$,
    $\delta=\min_{t \in \{1,2,\ldots,T\}} \delta_t>0$.
  (B3) $\sup_{t\in\{1,2,\ldots,T\}} \sup_{\btheta \in \Theta_n} 
       |Q_t(\btheta, \bX_n)-Q_t^*(\btheta)| \to_p 0$ as $n\to \infty$. 
  (B4) The penalty function $P_{\lambda_n}(\btheta)$ is non-negative and 
       converges to 0 uniformly over the set $\{\btheta_*^{(t)}: t=1,2,\ldots,T\}$ 
       as $n\to \infty$, where $\lambda_n$ is a regularization parameter and its 
       value can depend on the sample size $n$. 
  Let $\btheta_n^{(t)}=\arg\max_{\btheta\in \Theta_n} 
 \{ Q_t(\btheta, \bX_n)-P_{\lambda_n}(\btheta)\}$. Then the uniform convergence holds, i.e., 
  $\sup_{t \in \{1,2,\ldots,T\}} \|\btheta_n^{(t)}- \btheta_*^{(t)}\|\to_p 0$. 
 \end{lemma}

 \begin{proof} 
 Consider two events (i) $\sup_{t \in \{1,2,\ldots,T\} }
  \sup_{\btheta \in \Theta_n\setminus B_t(\epsilon)} 
  |Q_t(\btheta,\bX_n)-Q_t^*(\btheta)| < \delta/2$, and (ii) 
 $\sup_{t \in \{1,2,\ldots,T\} }$ 
  $\sup_{\btheta \in B_t(\epsilon)}
  |Q_t(\btheta,\bX_n)-Q_t^*(\btheta)| < \delta/2$. 
 From event (i), we can deduce that for any $t \in \{1,2,\ldots, T\}$ and any 
 $\btheta \in \Theta_n\setminus B_t(\epsilon)$, 
 $Q_t(\btheta, \bX_n) < Q_t^*(\btheta)+\delta/2 \leq Q_t^*(\btheta_*^{(t)}) -\delta_t 
  +\delta/2 \leq Q_t^*(\btheta_*^{(t)}) -\delta/2$. Therefore, 
  $Q_t(\btheta, \bX_n) -P_{\lambda_n}(\btheta) <  Q_t^*(\btheta_*^{(t)}) -\delta/2
   -o(1)$ by condition (B4).

 From event (ii), we can deduce that for any 
 $t \in \{1,2,\ldots, T\}$ and any $\btheta \in B_t(\epsilon)$, 
  $Q_t(\btheta, \bX_n)> Q_t^*(\btheta) -\delta/2$ and 
  $Q_t(\btheta_*^{(t)}, \bX_n)> Q_t^*(\btheta_*^{(t)}) -\delta/2$. 
  Therefore,  $Q_t(\btheta_*^{(t)}, \bX_n)-P_{\lambda_n}(\btheta_*^{(t)}) 
   > Q_t^*(\btheta_*^{(t)}) -\delta/2- o(1)$ by condition (B4).  
  
 If both events hold simultaneously, 
 then we must have ${\btheta}_n^{(t)} \in B_t(\epsilon)$ for all $t \in \{1,2,\ldots, T\}$ 
  as $n \to \infty$.  
 By condition (B3), the probability that both events hold tends to 1. 
 Therefore,
\[
P(\mbox{$\btheta_n^{(t)} \in B_t(\epsilon)$ for all $t=1,2,\ldots,T$}) \to 1,
\] 
which concludes the lemma. 
\end{proof}

 Theorem \ref{them1new} establishes the consistency of  
 ${\btheta}_n^{(t+1)}$ with respect to $\btheta_*^{(t)}$ for each $t=1,2,\ldots,T$. 
 However, in the small-$n$-large-$p$ scenario, 
 ${\btheta}_n^{(t+1)}$ is not well defined. 
  For this reason, a sparsity constraint needs to be imposed on $\btheta$. For example, we 
  can apply a regularization method to get an estimate of $\btheta_*^{(t)}$; 
 that is, we can define 
   \begin{equation} \label{mingeq5}
   \btheta_{n,p}^{(t+1)}
   =\arg\max_{\btheta\in \Theta_n}\left\{ \hat{G}_n(\btheta|\tbx,\btheta_n^{(t)})
    -P_{\lambda_n}(\btheta)\right\},
   \end{equation}  
   where the penalty function $P_{\lambda_n}(\btheta)$ constrains the sparsity 
   of the solution.  Assume that 
  \begin{itemize} 
  %\item[(A5)] For each $t \in \{1,2,\ldots,T\}$, 
  %  $\sum_{i=1}^p I(\theta_{*,i}^{(t)}\ne 0)=O(n^{1/2-\alpha})$ 
  %  for some constant $0<\alpha \leq 1/2$, where $\theta_{*,i}^{(t)}$ denotes the 
  %  $i$th element of $\theta_*^{(t)}$ and $I(\cdot)$ is the indicator function.  
  \item[(C1)] The penalty function $P_{\lambda_n}(\btheta)$ is non-negative,
   ensures the existence of $\btheta_{n,p}^{(t+1)}$ for all $n \in \mN$ 
   and $t=1,2,\ldots,T$,  
   and converges to 0 uniformly over the set  
   $\{\btheta_*^{(t)}: t=1,2\ldots,T\}$ as $n\to\infty$. 
  \end{itemize} 
  
 \begin{corollary} \label{cor0A} If the conditions (A1)-(A4) and (C1) hold, 
 then the regularization estimator $\btheta_{n,p}^{(t+1)}$ in (\ref{mingeq5}) 
 is uniformly consistent to $\btheta_*^{(t)}$ over $t=1,2,\ldots,T$, i.e.,  
 $\sup_{t \in\{1,2,\ldots,T\}} \|\btheta_{n,p}^{(t+1)}-\btheta_*^{(t)}\| \to_p 0$
 as $n\to \infty$.
 \end{corollary}
 \begin{proof} It follows the proof of Lemma \ref{lem00} directly. \end{proof}

   Take the high-dimensional regression as example. 
  If we allow $p$ to grow with $n$  at the rate 
   $p=O(n^\gamma)$ for some constant $\gamma>0$, allow the size of 
   $\bbeta_*^{(t)}$ for all $t$ to grow with $n$  at the rate  
    $O(n^{\alpha})$ for some constant $0<\alpha<1/2$, 
   choose $\lambda_n=O(\sqrt{\log(p)/n})$, 
   and set $P_{\lambda_n}(\btheta)=\lambda_n \sum_{i=1}^p c_{\lambda_n}(|\theta_i|)$,  
   where $c_{\lambda_n}(\cdot)$ is set in the form of SCAD (Fan and Li, 2001) 
   or MCP (Zhang, 2010) penalties, 
   then the condition (C1) is satisfied. For both the SCAD and MCP penalties, 
   $c_{\lambda_n}(|\theta_i|)=0$ if $\theta_i=0$ and bounded by a constant otherwise.  
   Similarly, if the beta-min assumption holds, i.e., there exists a 
   constant $\beta_{\min}>0$ such that $\min_{j \in S^*} |\beta_{*j}| \geq \beta_{\min}$, 
   where $S^*=\{j: \beta_{*j} \ne 0\}$ denotes the index set of non-zero   
   regression coefficients,  then the 
   reciprocal Lasso penalty (Song and Liang, 2015b) also satisfies (C1). 
   Note that, if $\Theta=\mR^p$, the Lasso penalty does not satisfy (C1) 
   as which is unbounded.  
   This explains why the Lasso estimate is unbiased even as $n\to\infty$.  
   However, if $\Theta_n$ is restricted to a bounded space, 
   then the Lasso penalty also satisfies (C1).

   Alternative to regularization methods, one may first restrict the space 
   of $\btheta_*^{(t)}$ to some low-dimensional subspace through 
   sure screening, and then find a consistent estimate in the subspace using 
   a conventional statistical methods, such as maximum likelihood, moment 
   estimation, or even regularization. 
   Both the $\psi$-learning  (Liang, Song and Qiu, 2015)
   sure independence screening (SIS) (Fan and Lv, 2008; Fan and Song, 2010) methods 
   belong to this class. 
   For $\psi$-learning, after correlation screening (based on the
   pseudo-complete data), the remaining network structure estimation procedure is essentially
   the same with the covariance selection method (Dempster, 1972) which,
   by nature, is a maximum likelihood estimation method.
   It is interesting to point out that the sure screening-based methods can 
   be viewed as a special subclass of regularization methods, for which the solutions
   in the low-dimensional subspace receives a zero penalty, and those outside
   the subspace receives a penalty of $\infty$. It is easy to see
   that such a binary-type penalty function satisfies condition (C1).

   Both the regularization and sure screening-based methods are constructive. 
   In what follows, we give a proof for the use of general 
   consistent estimation procedures in the IC algorithm. 
   Let $\btheta_{n,g}^{(t+1)}$ denote the estimate of $\btheta_*^{(t)}$ 
   produced by such a general consistent estimation procedure at iteration $t+1$.  
   Corollary \ref{cor0} shows that if $\btheta_{n,g}^{(t+1)}$ is 
   accurate enough for each $t$ (pointwisely) and the log-likelihood 
   function of the pseudo-complete data satisfies some moment conditions,  
   then the estimation procedure can be used in the IC algorithm. 
   Therefore, by its MLE nature in the subspace, the use of the 
   $\psi$-learning algorithm in the IC algorithm can also be justified 
   by Corollary \ref{cor0}. 

   %\item[(A5)] For each $t \in \{1,2,\ldots,T\}$, 
  %  $\sum_{i=1}^p I(\theta_{*,i}^{(t)}\ne 0)=O(n^{1/2-\alpha})$ 
  %  for some constant $0<\alpha \leq 1/2$, where $\theta_{*,i}^{(t)}$ denotes the 
  %  $i$th element of $\theta_*^{(t)}$ and $I(\cdot)$ is the indicator function.  

\begin{itemize}
\item[(C2)] [Conditions for general consistent estimate $\btheta_{n,g}^{(t)}$] 
 Assume that for each $t=1,2,\ldots, T$, 
  $\btheta_{n,g}^{(t+1)}-\btheta_*^{(t)}=O_p(1/\sqrt{n})$ (pointwisely) and  
 the Hessian matrix $\partial^2 G_n(\btheta|\tbx,\btheta_n^{(t)})/\partial \btheta \partial \btheta'$ is bounded 
 in a neighborhood of $\btheta_*^{(t)}$; let 
  $$Z_{t,i}'=\log f(\xobs_i,\txmis_i|\btheta_{n,g}^{(t+1)})-\int \log f(\xobs,\txmis|\btheta_{n,g}^{(t+1)})
  f(\xobs|\btheta^*) h(\txmis|\xobs_i, \btheta_n^{(t)})d\txmis d\xobs,$$ then  
  $E|Z_{t,i}'|^m \leq m! \tilde{M}_b^{m-2} \tilde{v}_i/2$ for every $m\geq 2$ and some constants 
  $\tilde{M}_b>0$ and $\tilde{v}_i=O(1)$.
 % That is, each $Z_{t,i}'$ is  a sub-exponential random variable. 
 \end{itemize}

 \begin{corollary} \label{cor0} Assume (A1)-(A4) and (C2). Then  $\btheta_{n,g}^{(t+1)}$ is 
  uniformly consistent to $\btheta_*^{(t)}$ over $t=1,2,\ldots,T$, i.e., 
 $\sup_{t \in\{1,2,\ldots,T\}} \|\btheta_{n,g}^{(t+1)}-\btheta_*^{(t)} \| \to_p 0$
 as $n\to \infty$.
 %  $\sup_{t\in\{1,2,\ldots,T\}} P(|\btheta_{n,g}^{(t+1)} -\btheta_*^{(t)}|>\epsilon) \to 0$ as $n\to \infty$.  
  \end{corollary}  
 \begin{proof}  
 Applying Taylor expansion to $G_n(\btheta|\btheta_n^{(t)})$ at $\btheta_*^{(t)}$, we get 
 $G_n(\btheta_{n,g}^{(t+1)}|\btheta_n^{(t)}) - G_n(\btheta_{*}^{(t)}|\btheta_n^{(t)}) 
 =O_p(1/n)$, following from condition (C2) and condition (A4) that $G_n(\btheta|\btheta_n^{(t)})$ 
 is maximized at $\btheta_*^{(t)}$. Therefore, 
\[
\begin{split} 
 n[\hat{G}_n(\btheta_{n,g}^{(t+1)}|\tbx,\btheta_n^{(t)}) - G_n(\btheta_{*}^{(t)}|\btheta_n^{(t)})]
  &=Z_{t,1}'+\cdots+Z_{t,n}'+n[G_n(\btheta_{n,g}^{(t+1)}|\btheta_n^{(t)})- 
  G_n(\btheta_{*}^{(t)}|\btheta_n^{(t)})]\\
 & = Z_{t,1}'+\cdots+Z_{t,n}'+ \epsilon_n, \\
\end{split}
 \] 
 where $\epsilon_n=O_p(1)$, and  
 \begin{equation} \label{Goodeq3}
P(n|\hat{G}_n(\btheta_{n,g}^{(t+1)}|\tbx,\btheta_n^{(t)}) - G_n(\btheta_{*}^{(t)}|\btheta_n^{(t)})|>nz) 
 \leq P(|Z_{t,1}'+\cdots+Z_{t,n}'|>nz-|\epsilon_n|). 
\end{equation} 
 By Bernstein's inequality, 
 \begin{equation} \label{Goodeq1}
 P(|Z_{t,1}'+\cdots+Z_{t,n}|>nz-|\epsilon_n|) \leq 2 \exp\left\{-\frac{1}{2} \frac{(z-|\epsilon_n|/n)^2}{
   \tilde{v}'+\tilde{M}_b'(z-|\epsilon_n|/n)} \right\}, 
 \end{equation} 
 for $\tilde{v}' \geq (\tilde{v}_1+\cdots+\tilde{v}_n)/n^2$ and $\tilde{M}_b'=\tilde{M}_b/n$. 
 Applying Taylor expansion to the 
 right of (\ref{Goodeq1}) at $z$ and combining with (\ref{Goodeq3}) leads to 
 \begin{equation} \label{Goodeq4}
 P(|\hat{G}_n(\btheta_{n,g}^{(t+1)}|\tbx,\btheta_n^{(t)}) - G_n(\btheta_{*}^{(t)}|\btheta_n^{(t)})|>z)
  \leq  K \exp\left\{-\frac{1}{2} \frac{z^2}{\tilde{v}'+\tilde{M}_b' z} \right\}, 
%  := K \exp\left\{-\frac{1}{2} \frac{z^2}{\tilde{v}+\tilde{M}_b z} \right\},  
 \end{equation}
 where $K=2+\frac{3}{\tilde{M}_b'}O_p(1/n)=2+\frac{3}{\tilde{M}_b}O_p(1)$, since 
 the derivative $|d[z^2/(\tilde{v}'+\tilde{M}_b' z)]/dz| \leq 3/\tilde{M}_b'$.
% Since $K$ will converge to 2 as $n\to \infty$,
% the above proof implies that Bernstein's inequality also holds for asymptotically 
% centered variables. 

 As in the proof of Theorem \ref{them0}, by applying Lemma 2.2.10 of van der Vaart and Wellner (1996), 
 we can prove    
\begin{equation} \label{penneq0} 
 \sup_{\btheta_n^{(t)} \in \Theta_n, t \in \{1,2,\ldots,T\}} \left|\hat{G}_n(\btheta_{n,g}^{(t+1)}|\tbx, 
  \btheta_n^{(t)})- G_n(\btheta_*^{(t)} |\btheta_n^{(t)}) \right| \to_p 0.
\end{equation} 
 Note that, as implied by the proof of Lemma 2.2.10 of van der Vaart and Wellner (1996), 
 (\ref{penneq0}) holds for a general constant $K$ in (\ref{Goodeq4}). 
 Then, by condition (A4), we must have the uniform convergence that 
 $\btheta_{n,g}^{(t+1)} \in B_t(\epsilon)$ for all $t$ as $n\to \infty$, 
 where $B_t(\epsilon)$ is as defined in (A4).  
 This statement can be proved by contradiction as follows:
 
 Assume $\btheta_{n,g}^{(i+1)} \notin B_{i}(\epsilon)$ for some $i\in\{1,2,\ldots,T\}$. By  
 the uniform convergence established in Theorem \ref{them0}, 
  $\left|\hat{G}_n(\btheta_{n,g}^{(i+1)}|\tbx, 
  \btheta_n^{(i)})- G_n( \btheta_{n,g}^{(i+1)} |\btheta_n^{(i)}) \right| =o_p(1)$.
 Further, by condition (A4) and the assumption $\btheta_{n,g}^{(i+1)} \notin B_{i}(\epsilon)$, 
 \[
\begin{split}
 \left| \hat{G}_n(\btheta_{n,g}^{(i+1)}|\tbx, 
  \btheta_n^{(i)})- G_n(\btheta_*^{(i)} |\btheta_n^{(i)}) \right| & \geq  
  \left|G_n( \btheta_{n,g}^{(i+1)} |\btheta_n^{(i)}) - G_n(\btheta_*^{(i)} |\btheta_n^{(i)}) \right| - 
  \left|\hat{G}_n(\btheta_{n,g}^{(i+1)}|\tbx, 
  \btheta_n^{(i)})- G_n( \btheta_{n,g}^{(i+1)} |\btheta_n^{(i)}) \right| \\
  &  \geq \delta-o_p(1), 
 \end{split}
 \]
 which contradicts with the uniform convergence established in (\ref{penneq0}).  
 This concludes the proof. 
 \end{proof}  

\noindent {\bf Remark R3} {\it (On the accuracy of $\btheta_{n,g}^{(t)}$'s)}
   Condition (C2) restricts the consistent estimates to those having a distance
   to the true parameter point of the order  $O_p(1/\sqrt{n})$. Such
   condition can be satisfied by some estimation procedures   
   in the low-dimensional subspace, e.g., maximum likelihood, 
   for which both the variance and bias are often of the order $O(1/n)$ (Firth, 1993) 
   and therefore the root mean squared error is of the order $O(1/\sqrt{n})$. 
   To make the result of Corollary \ref{cor0} more general to include more 
   estimation procedures, we can relax this order 
   to $\btheta_{n,g}^{(t+1)}-\btheta_*^{(t)}=O_p(n^{-1/4})$, if we would like to relax 
   the order of $T$ to $\log(T)=o(\sqrt{n})$ and the order of metric entropy to 
   $\log N(\epsilon,\mathcal{G}_{n,M},L_1(\mathbb{P}_n))=o_p^*(\sqrt{n})$. 
   As mentioned in remarks (R1) and (R2), both the order of $T$ and the order 
   of metric entropy are technical conditions and relaxing them to the order 
   of $O(\sqrt{n})$ will not restrict much the applications of the IC algorithm.   
   The proof for this relaxation is straightforward, following the proof of 
   Corollary \ref{cor0}.

\paragraph{2. Proof of ergodicity of the Markov chain $\{\btheta_n^{(t)} \}$ }

 Although the IC algorithm is different from the stochastic EM algorithm in 
 the $\btheta_n^{(t)}$-updating step, the Markov chains $\{\btheta_n^{(t)}\}$  
 induced by the two algorithms share some similar properties as well as 
 similar proofs. 
 The following two lemmas,  Lemma \ref{lem1} and Lemma \ref{lem2},  
 can be proved in the same way as in Nielsen (2000), and thus the proofs 
 are omitted. 

\begin{lemma} \label{lem1}
 The Markov chain $\{\btheta_n^{(t)}\}$ is irreducible and aperiodic.
\end{lemma}

\begin{lemma} \label{lem2}
If (A1) holds, then the Markov chain $\{\btheta_n^{(t)} \}$ has the weak Feller property,
 and any compact subsets of $\Theta$ are small.
\end{lemma}

%\begin{lemma} \label{DCT} ({\it Dominated Convergence Theorem}) If for some random variable $Z$, 
% $|X_n| \leq |Z|$ for all $n$ and $E|Z|<\infty$, then $X_n \stackrel{d}{\to} X$ implies that 
% $E X_n \to E X$. 
%\end{lemma} 
%The proof of this lemma is based on the Skorohod Representation Theorem and can be 
%found in many textbooks. 
 
% \begin{proof} Fatou's lemma states that 
% \[
%  E \lim_n \inf |X_n| \leq \lim_n \inf E|X_n|.
%  \]
% A second application of Fatou's lemma to the nonnegative random variable $|Z|-|X_n|$ implies 
% \[
%  E|Z|-E \lim_n\sup |X_n| \leq E |Z|- \lim_n\sup E|X_n|.
%  \]
%  Because $E|Z|<\infty$, subtracting $E|Z|$ preserves this inequality, so we obtain 
%  \[
%  \lim_n\sup E|X_n| \leq E \lim_n\sup |X_n|.
%  \]
%  Therefore, 
% \[
%   E \lim_n \inf |X_n| \leq \lim_n \inf E|X_n| \leq  \lim_n\sup E|X_n| \leq E \lim_n\sup |X_n|.
% \]
% Hence, the proof would be complete if $|X_n| \stackrel{a.s.}{\to} |X|$. This is where we invoke the 
% Skorohod Representation Theorem: Because there exists a sequence $Y_n$ that does converge almost 
% surely to $Y$, having the same distributions and expectations as $X_n$ and $X$, the above 
%  argument shows that $E Y_n \to E Y$, hence $E X_n \to E X$, completing the proof.
% \end{proof}  

If (A1) holds and $\Theta_n$ is restricted to a compact set,
 then the Markov chain $\{\btheta_n^{(t)}\}$ is ergodic.
Here we would like to establish the ergodicity of 
the Markov chain $\{\btheta_n^{(t)}\}$ under a more general scenario $\Theta_n=\mR^p$. 
This can be done by verifying a drift condition. 
Similar to Nielsen (2000), we choose the negative log-likelihood function 
 of the observed data as the drift function, motivated by the drift in the EM 
 algorithm towards high-density areas. 

\begin{theorem} \label{them1}
If (A1)--(A3) hold, then $\{\btheta_n^{(t)} \}$ is almost surely
 ergodic for sufficiently large $n$.
\end{theorem}

\begin{proof}
Let $\upsilon(\btheta)=C-\frac{1}{n}\log f(\xobs_{1},\dots,\xobs_{n}|\btheta)$, where $C$ denotes a constant such that 
 $\upsilon(\btheta) \geq 0$ for all $\btheta \in \Theta_n$.  
Since $\upsilon(\btheta)$ is 
nonnegative, it can be used to build the drift condition. Define 
\[
\begin{split}
\Delta\upsilon(\btheta)&=E_h[\upsilon(\btheta_{n}^{(t+1)})-\upsilon(\btheta_{n}^{(t)})]
=E_h[\frac{1}{n}\log f(\bxobs|\btheta_{n}^{(t)})-\frac{1}{n}\log f(\bxobs|\btheta_n^{(t+1)})] \\ 
&=E_h[\frac{1}{n}\log f(\tilde{\bx}|\btheta_n^{(t)})-\frac{1}{n}\log f(\tilde{\bx}|\btheta_{n}^{(t+1)})]
-E_h[\frac{1}{n}\log h(\tbxmis|\bxobs,\btheta_{n}^{(t)})-\frac{1}{n}\log h(\tbxmis|\bxobs,\btheta_{n}^{(t+1)})] \\
&= (I) +(II),
\end{split}
\]
where $E_h$ refers to the expectation with respect to the predictive distribution $h(\tbxmis|\bxobs,\btheta_n^{(t)})$. 

First, we consider the negative of part (I), which can be decomposed as 
\[
\begin{split}
-(I) &= E_h[ \frac{1}{n}\log f(\tilde{\bx}|\btheta_{n}^{(t+1)}) - \frac{1}{n}\log f(\tilde{\bx}|\btheta_n^{(t)})] \\
     &= E_h[ \frac{1}{n}\log f(\tilde{\bx}|\btheta_{n}^{(t+1)}) -
         \frac{1}{n}\log f(\tilde{\bx}| \tM(\btheta_{n}^{(t)})) ]  
     + E_h[  \frac{1}{n}\log f(\tilde{\bx}| \tM(\btheta_{n}^{(t)})) 
       - \frac{1}{n}\log f(\tilde{\bx}|\btheta_n^{(t)}) ], \\ 
\end{split}
\]
where the function $M(\btheta)$ is defined by 
\begin{equation} \label{mapeq}
\tM(\btheta)=\arg\max_{\btheta'} E_{\btheta} \log f(\tilde{\bx}|\btheta')=\arg\max_{\btheta'} 
 \int f(\xobs,\txmis|\btheta') f(\xobs|\btheta^*) h(\txmis|\xobs,\btheta)d\txmis d\xobs.
\end{equation}
 By the ULLN established in Theorem \ref{them0}, we have 
 \[
 \frac{1}{n}\log f(\tilde{\bx}|\btheta_{n}^{(t+1)}) -
   \frac{1}{n}\log f(\tilde{\bx}| \tM(\btheta_{n}^{(t)})) \to_p  
  G_n(\btheta_n^{(t+1)}|\btheta_n^{(t)}) - G_n(M(\btheta_n^{(t)})|\btheta_n^{(t)}). 
 \] 
 From Theorem \ref{them1new}, we have 
  $\btheta_{n}^{(t+1)} - \tM(\btheta_{n}^{(t)})  \to_p 0$. Further, by the continuity of 
 $G_n(\btheta|\btheta')$ with respect to $\btheta$, we have 
  $G_n(\btheta_n^{(t+1)}|\btheta_n^{(t)}) - G_n(M(\btheta_n^{(t)})|\btheta_n^{(t)}) \to_p 0$ and 
 thus  $\frac{1}{n}\log f(\tilde{\bx}|\btheta_{n}^{(t+1)}) -
  \frac{1}{n}\log f(\tilde{\bx}| \tM(\btheta_{n}^{(t)})) \to_p  0$. 
 Then, by the boundedness of $\log f(\tilde{\bx}|\btheta)$ (condition (A2)) 
  and the dominated convergence theorem, 
 \begin{equation} \label{proofeq11}
 E \left[ \frac{1}{n}\log f(\tilde{\bx}|\btheta_{n}^{(t+1)}) -
   \frac{1}{n}\log f(\tilde{\bx}| \tM(\btheta_{n}^{(t)})) \right] \to 0,
 \end{equation} 
 where the expectation is with respect to the joint density function of $\tilde{\bx}=(\bxobs,\tbxmis)$.  
 Note that for any $\btheta \in \Theta_n$, we have 
 \begin{equation} \label{proofeq14}  
  E_h [ \frac{1}{n} \log f(\tilde{\bx}|\btheta) ] =\frac{1}{n} \sum_{i=1}^n E_h \log f(\tilde{x_i} |\btheta) 
  \stackrel{\Delta}{=} \frac{1}{n} \sum_{i=1}^n g(\xobs_i),
 \end{equation}
 where $g(\xobs_i)$'s are mutually independent, but not necessarily identically distributed 
 due to the presence of missing data. 
 Then, by (\ref{proofeq11}), (A2) and Kolmogorov's SLLN, we have 
 \begin{equation} \label{aseq0}
  E_h \left[ \frac{1}{n}\log f(\tilde{\bx}|\btheta_{n}^{(t+1)}) -
         \frac{1}{n}\log f(\tilde{\bx}| \tM(\btheta_{n}^{(t)})) \right ] \to 0, \quad \mbox{a.s.},
 \end{equation}
  as $n \to \infty$. 
 Therefore, there exists a constant $c>0$ and a large number $N$ such that 
 \begin{equation} \label{proofeq12}
  -c< E_h\left[ \frac{1}{n}\log f(\tilde{\bx}|\btheta_{n}^{(t+1)}) -
   \frac{1}{n}\log f(\tilde{\bx}| \tM(\btheta_{n}^{(t)})) \right] < c, \quad \mbox{a.s.},
 \end{equation}
 for any $n>N$ and any $t>0$.   
 With a similar argument to (\ref{proofeq14}), by invoking Kolmogorov's SLLN, 
 it can be shown that there exists a constant $\delta>0$ such that 
\begin{equation}\label{proofeq13}
E_h[\frac{1}{n}\log f(\tilde{\bx}| \tM(\btheta_{n}^{(t)})) - \frac{1}{n}\log f(\tilde{\bx}|\btheta_n^{(t)})]  \to  \delta, \quad \mbox{a.s.},
\end{equation}
 for any $t>0$ as $n \to \infty$. 
 Combining (\ref{proofeq12}) and (\ref{proofeq13}), we have  $-c-\delta <  (I) < c$ holds almost surely
 for sufficiently large $n$.  

 Next, by Jensen's inequality, we have 
 \[
 \begin{split} 
  (II) &= E_h \left[\frac{1}{n}\log h(\tbxmis|\bxobs,\btheta_{n}^{(t+1)})-\frac{1}{n}
           \log h(\tbxmis|\bxobs,\btheta_{n}^{(t)}) \right ] 
        \leq \frac{1}{n} \log E_h\left( \frac{ h(\tbxmis|\bxobs,\btheta_{n}^{(t+1)})}{
           h(\tbxmis|\bxobs,\btheta_{n}^{(t)}) } \right)  \\
       & = \frac{1}{n} \log \int h(\tbxmis|\bxobs,\btheta_{n}^{(t+1)}) d \tbxmis =0. \\
 \end{split} 
 \]
 Combining the results of (I) and (II), we have that 
 $\Delta\upsilon(\btheta) < c$ almost surely for all $\btheta \in \Theta_n$. 
 Choose $b$ as a positive number less than $c+\delta$ and $D$ as a compact set including 
 $\{\btheta \in \Theta_{n}: \Delta \upsilon(\boldsymbol{\theta}) \in [-b,c)\}$. 
 In summary, we have
\[
\Delta \upsilon(\boldsymbol{\theta})\leq
\begin{cases}
  c,  & \btheta\in D, \\
  -b,    & \btheta\in \Theta_n\setminus D, \\
\end{cases}
\]
almost surely. Hence, the strict drift condition $V_2$ (Meyn and Tweedie, 2009, p263)
 is almost surely satisfied. 

Since $(\btheta_{n}^{(t)})_{t \in \mathbb{N}_{0}}$ also has weak Feller property (see Lemma \ref{lem2}), 
 we can further conclude that an invariant probability measure $\pi$ almost surely exists for this Markov chain 
 (Meyn and Tweedie, 2009, Theorem 12.3.4).
Since $(\btheta_{n}^{(t)})_{t \in \mathbb{N}_{0}}$ is irreducible (shown in Lemma \ref{lem1}), $D$ is
a compact set and thus a small set (shown in Lemma \ref{lem2}), and the drift condition $V_2$ is stronger than 
 the drift condition $V_1$ (Meyn and Tweedie, 2009, p189), 
 we can show that $(\btheta_{n}^{(t)})_{t \in \mathbb{N}_{0}}$ is Harris recurrent 
 (Meyn and Tweedie, 2009, Theorem 9.1.8). 
Since $(\btheta_{n}^{(t)})_{t \in \mathbb{N}_{0}}$ is irreducible and has an invariant probability measure $\pi$, 
 it is also a positive chain (Meyn and Tweedie, 2009, p 230). Therefore, it is a positive Harris recurrent 
chain (Meyn and Tweedie, 2009, p 231).
Finally, since $(\btheta_{n}^{(t)})_{t \in \mathbb{N}_{0}}$ is aperiodic (shown in Lemma \ref{lem1}) and 
 positive Harris recurrent, we can conclude that it is almost surely ergodic (Meyn and Tweedie, 2009, Theorem 13.3.3). 
\end{proof}

\paragraph{3. Proof of consistency of the IC estimator} 
 
To prove the consistency of the IC estimator, we consider the mapping defined in (\ref{mapeq}). 
 For the C-step, we have $\btheta_*^{(t)}=M(\btheta_n^{(t)})$. 
 Also, $\btheta^*$, the true value of $\btheta_n$, is a fixed point of the mapping. 
 Further, to show that the mean of the stationary distribution of the Markov chain
 forms a consistent estimate of $\btheta^*$, we make the following assumption.

\begin{itemize}
 \item[(A5)] The mapping $M(\btheta)$ is
                differentiable.  Let $\lambda_n(\btheta)$ be the largest singular value of
                $\partial M(\btheta)/\partial \btheta$.
                There exists a number $\lambda^* <1$
                such that $\lambda_n(\btheta) \leq \lambda^*$ for all $\btheta \in \Theta_n$
                for sufficiently large $n$ and almost every $\bxobs$-sequence.
 \end{itemize}

\noindent {\bf Remark R4} {\it (On contraction mapping)}  
% Let $\|\cdot\|$ denote the Euclidean-norm.  
 The condition (A5) directly implies
 \begin{equation} \label{mapeq2}
 \|M(\btheta_n^{(t)})-\btheta^*\| = \|M(\btheta_n^{(t)})-M(\btheta^*)\| \leq \lambda^* \|\btheta_n^{(t)}-\btheta^*\|,
 \end{equation}
 that is, the mapping is a contraction. We note that a continuous application
 of the mapping, i.e., setting $\btheta_n^{(t+1)}=\btheta_*^{(t)}=M(\btheta_n^{(t)})$ for all $t$, leads to
 a monotone increase of the expectation $E_{\btheta_n^{(t)}} \log f_{\btheta}(\tilde{\bx})$.
 Since $E_{\btheta_n^{(t)}} \log f_{\btheta}(\tilde{\bx})$ attains its maximum at
 $E_{\btheta^*} \log f_{\btheta^*}(\tilde{\bx})$, it is reasonable to assume that
 $M(\btheta_n^{(t)})$ is closer to $\btheta^*$ than $\btheta_n^{(t)}$. This condition should hold
 for sufficiently large $n$, at which $\btheta_*^{(t)}$'s and $\btheta^*$ are all unique
 as assumed in condition (A4).
 We note that a similar contraction condition has been used in analysis of the SEM algorithm (Proposition 3, Nielsen, 2000).
 Some other conditions can potentially be specified based on the fixed-point theory (see e.g., Khamsi and Kirk, 2001).

\begin{theorem} \label{them2}
Assume (A1)-(A5) and $\sup_{n,t} E\|\btheta_n^{(t)}\| <\infty$. 
 Then for sufficiently large $n$, sufficiently large $t$,
 and almost every  $\bxobs$-sequence,
$ \| \btheta_n^{(t)} - \btheta^* \|=o_p(1)$.
 Furthermore, the sample average of the Markov chain forms a consistent estimate of $\btheta^*$, i.e.,
 $\| \frac{1}{T} \sum_{t=1}^T \btheta_n^{(t)} - \btheta^* \|=o_p(1)$,
 as $n \to \infty$ and $T\to \infty$.
\end{theorem}
\begin{proof}
 By Theorem \ref{them1}, the Markov chain $\{\btheta_n^{(t)}\}$ converges to a stationary distribution.
 For simplicity, we suppress the supscript $t$, let $\btheta_n$ denote the current 
 sample, and let $\btheta_n'$ denote the next iteration sample. 
 Therefore, $\|\btheta_n' - \btheta^*\| \leq \|\btheta_n' - \tM(\btheta_n) \| 
  + \| \tM(\btheta_n) - \btheta^* \|  
   \leq \|\btheta_n' - \tM(\btheta_n) \| +  \lambda^* \|\btheta_n-\btheta^*\|$, 
 where the last inequality follows from (\ref{mapeq2}). 
 Taking expectation on both sides leads to 
 \begin{equation} \label{proofeq1}
 E \|\btheta_n' - \btheta^*\|  \leq E \|\btheta_n' - \tM(\btheta_n) \| +  \lambda^* E \|\btheta_n-\btheta^*\| 
      \leq \frac{1}{1-\lambda^*} E \|\btheta_n' - \tM(\btheta_n) \|  
      = \frac{1}{1-\lambda^*} o(1)=o(1),
 \end{equation} 
 where the second inequality follows from the stationarity of the Markov chain, 
 and the first equality follows from Theorem \ref{them1new} and the existence of 
 $E \|\btheta_n\|$. 
 Finally, by Markov's inequality, we conclude the consistency of $\btheta_n^{(t)}$ as 
 an estimator of $\btheta^*$. 
 
 By (\ref{proofeq1}), we have $\|E (\btheta_n) -\btheta^* \| \leq E \|\btheta_n-\btheta^*\| =o(1)$,
 which implies that the mean of the stationary distribution of $\{\btheta_n^{(t)} \}$ converges to 
 $\btheta^*$ for sufficiently large $n$. Further, by the ergodicity of the 
 Markov chain  $\{\btheta_n^{(t)} \}$, we conclude the proof. 
\end{proof} 

\begin{corollary} \label{corMCMC1} 
 Assume (A1)-(A5), $\sup_{n,t} E\|\btheta_n^{(t)}\| <\infty$, 
 $h(\btheta)$ is a Lipschitz function on $\Theta_n$, and $\sup_{n,t}$ 
 $E\|h(\btheta_n^{(t)})\|<\infty$. 
 Then for sufficiently large $n$, sufficiently large $t$,
 and almost every  $\bxobs$-sequence,
$ \| h(\btheta_n^{(t)}) - h(\btheta^*) \|=o_p(1)$.
 Furthermore,  $\| \frac{1}{T} \sum_{t=1}^T h(\btheta_n^{(t)}) - h(\btheta^*) \|=o_p(1)$,
 as $n \to \infty$ and $T\to \infty$. 
\end{corollary} 
\begin{proof} The proof follows from the definition of Lipschitz function and 
   the proof of Theorem \ref{them2}. 
\end{proof}

\paragraph{4. Proof of ergodicity of the Markov chain for the ICC algorithm }

\begin{theorem} \label{them3}
 If (A1)-(A3) hold, the Markov chain $\{(\btheta_n^{(t,1)}, \ldots, \btheta_n^{(t,k)})\}$ is almost surely
 ergodic for sufficiently large $n$.
\end{theorem}

This theorem can be proved in a similar way to Theorem \ref{them1} with the detail 
given in the Supplementary Material.

\paragraph{5. Proof of consistency of the ICC estimator}

\begin{itemize}
\item[(A5$'$)] Let $M_i$ denote the mapping of the $i$th part of the CC-step, i.e.,
     $\btheta_*^{(t,i)}=M_i(\btheta_n^{(t+1,1)}, \ldots, \btheta_n^{(t+1,i-1)},$
     $\btheta_n^{(t,i)}, \ldots, \btheta_n^{(t,k)})$.
     Let $M=M_k\circ M_{k-1}\circ \ldots\circ M_1$ denote the joint mapping of $M_1,\ldots, M_k$.
     Let $\lambda_n(\btheta)$ denote the largest singular value of $\partial M(\btheta)/\partial \btheta$.
     There exists a number $\lambda^*<1$ such that
     $\lambda_n(\btheta) \leq \lambda^*$ for all $\btheta \in \Theta_n$,
     all sufficiently large $n$, and almost every $\bxobs$-sequence.
\end{itemize}

 This condition is reasonable: It is easy to see that a continuous application
 of the mapping $M$, i.e., applying $M_i$'s in a circular manner, leads to
 a monotone increase of the function $E_{\btheta_n^{(t)}} \log f_{\btheta}(\tilde{\bx})$.
 Similar to Theorem \ref{them2}, we can
 prove the following theorem with the detail given in 
 the Supplementary Material.

\begin{theorem} \label{them4}
Assume (A1)-(A4), (A5$'$) and $\sup_{n,t} E|\btheta_n^{(t)}| <\infty$.
Then for sufficiently large $n$, sufficiently large $t$,
 and almost every $\bxobs$-sequence, $\| \btheta_n^{(t)} - \btheta^* \|=o_p(1)$.
 Furthermore, the sample average of the Markov chain also forms a consistent estimate of $\btheta^*$, i.e.,
$ \| \frac{1}{T} \sum_{t=1}^T \btheta_n^{(t)} - \btheta^* \|=o_p(1)$,
 as $n\to \infty$ and $T\to \infty$.
 \end{theorem}

\begin{corollary} \label{corMCMC2}
 Assume (A1)-(A4), (A5$'$), $\sup_{n,t} E\|\btheta_n^{(t)}\| <\infty$,
 $h(\btheta)$ is a Lipschitz function on $\Theta_n$, and $\sup_{n,t}$
 $E\|h(\btheta_n^{(t)})\|<\infty$.
 Then for sufficiently large $n$, sufficiently large $t$,
 and almost every  $\bxobs$-sequence,
$ \| h(\btheta_n^{(t)}) - h(\btheta^*) \|=o_p(1)$.
 Furthermore, $\| \frac{1}{T} \sum_{t=1}^T h(\btheta_n^{(t)}) - h(\btheta^*) \|=o_p(1)$,
 as $n \to \infty$ and $T\to \infty$.
\end{corollary}
\begin{proof} The proof follows from the definition of Lipschitz function and
   the proof of Theorem \ref{them4}.
\end{proof}

% \begin{proof}
% Recall that $M$ has been defined as the joint mapping of $M_1, \ldots, M_k$.  
% Then $M(\btheta_n^{(t)})=M_k\circ M_{k-1} \circ \ldots \circ M_1(\btheta_n^{(t)})
% = (\btheta_*^{(t,1)}, \dots, \btheta_*^{(t,k)})$.  Note that $\btheta^*$ is still 
% a fixed-point of $M$, i.e., $\btheta^*=M(\btheta^*)$. 
% From part (ii) of Theorem \ref{them0}, $\btheta_n^{(t+1,i)}$ forms a consistent estimator of 
% $\btheta_*^{(t,i)}$ for each $i=1,\ldots, k$. Since $k$ is finite, we have 
% \[
%  \btheta_n^{(t+1)} -M(\btheta_n^{(t)}) \to 0, \quad \mbox{in probability},
% \]
% for sufficiently large $n$.  With the same arguments as in the proof of Theorem \ref{them2}, we 
% can show that 
% \[
%  E \| \btheta_n^{(t+1)} -M(\btheta_n^{(t)}) \| \to 0. 
% \]
% Since an attraction property has been assumed for the joint mapping $M$, the following results hold for 
% the ICC algorithm as $t\to \infty$ and $n\to \infty$: 
% \[
% \| \btheta_n^{(t)} - \btheta^* \|=o_p(1),
% \]
%  and 
% \[
% \| \frac{1}{T} \sum_{t=1}^T \btheta_n^{(t)} - \btheta^* \|=o_p(1),
% \]
% following the same arguments as in Theorem \ref{them2}.
% \end{proof}

 \section*{References}
 
 \begin{description}

% \item[] Barab\'asi, A. and Albert, R. (1999). Emergence of scaling in random networks.
%   {\it Science}, {\bf 286}, 509.

\item[] Besag, J. (1974). Spatial Interaction and the Statistical Analysis of Lattice Systems (with discussion). 
    \JRSSB, {\bf 36}(2), 192-225. 

% \item[] Billingsley, P. (1986). {\it Probability and Measure} (2nd Edition). John Wiley \& Sons: New York.

 \item[] Bo, T.H., Dysvik, B. and Jonassen, I. (2004). LSimpute: accurate estimation of missing values in 
         microarray data with least square methods. {\it Nucleic Acids Research}, {\bf 32}:e34. 

\item[]  Buuren, S. and Groothuis-Oudshoorn, K. (2011). mice: Multivariate imputation by chained equations in R. 
         {\it Journal of statistical software}, {\bf 45}(3).

% \item[] Bowden, R. (1973). The theory of  parametric identification. {\it Econometrica}, {\bf 41}, 1069-1074.

 \item[] Cai, J.-F., Cand\`es, E. and Shen, Z. (2010). A singular value thresholding algorithm for matrix 
         completion. {\it SIAM Journal on Optimization}, {\bf 20}, 1956-1982. 

\item[] Castillo, I., Schmidt-Hieber, J. and van der Vaart, A.W. (2015). Bayesian linear regression with 
       sparse priors. \ANNALS, {\bf 43}, 1986-2018.

\item[] Celeux, G. and Diebolt, J. (1985). The SEM algorithm: a probabilistic teacher algorithm derived from the 
        EM algorithm for the mixture problem. {\it Computational Statistics Quarterly}, {\bf 2}, 73-82. 

% \item[] Chen, X. and Xie, M. (2014). A split-and-conquer approach for analysis of extraordinarily large data.
%         {\it Statistica Sinica}, {\bf 24}, 1655-1684.

 \item[] Dempster, A.P. (1972). Covariance selection. {\it Biometrics}, 28, 157-175.

 \item[] Dempster, A.P., Laird, N., and Rubin, D.B. (1977). Maximum likelihood from incomplete data 
         via the EM algorithm. \JRSSB, {\bf 39}, 1-38. 

\item[] Dobra, A., Hans, C., Jones, B., Nevins, J.R., Yao, G., and West, M. (2004).
         Sparse graphical models for exploring gene expression data. {\it Journal of Multivariate Analysis},
         {\bf 90}, 196-212.

\item[] Efron, B. (2004). Large-scale  simultaneous  hypothesis  testing: The choice of a null hypothesis. 
 \JASA, {\bf 99}, 96-104. 

% \item[] Fan, J., Feng, Y., Saldana, D.F., Samworth, R., and Wu, Y. (2016). Sure Independence Screening: 
%         R Package {\it SIS}. Downloadable at https://cran.r-project.org/web/packages/SIS. 

\item[] Fan, J. and Li, R. (2001). Variable selection via nonconcave penalized likelihood and its oracle properties.
   \JASA, {\bf 96}, 1348-1360.

\item[] Fan, J. and Lv, J. (2008). Sure Independence Screening for Ultrahigh Dimensional Feature Space.
        \JRSSB, {\bf 70}, 849-911.

\item[] Fan, J. and Song, R. (2010). Sure independence screening in generalized linear model with NP-dimensionality.
        \ANNALS, {\bf 38}, 3567-3604. 

\item[] Fan, Y. and Li, R. (2012). Variable selection in linear mixed effects models. 
        {\it The Annals of Statistics}, {\bf 40}, 2043-2068.

\item[] Firth, D. (1993). Bias reduction of maximum likelihood estimates. {\it Biometrika}, {\bf 80}(1), 27-38. 

\item[] Friedman, J., Hastie, T. and Tibshirani, R. (2008). Sparse inverse covariance estimation
 with the graphical lasso. {\it Biostatistics}, {\bf 9}, 432-441.

\item[] Garcia, R.I., Ibrahim, J.G. and Zhu, H. (2010). Variable selection for regression models 
        with missing data. {\it Statistica Sinica}, {\bf 20}, 149-165. 

 \item[] Gasch, A.P., Spellman, P.T., Kao, C.M., Carmel-Harel, O., Eisen, M.B., Storz, G., Botstein, D., 
         and Brown, P.O. (2000). Genomic expression programs in the response of yeast cells to environmental changes. 
         {\it Molecular Biology of the Cell}, {\bf 11}, 4241-4257. 

\item[] Gelman, A. and Rubin, D. B. (1992). Inference from iterative simulation using multiple sequences. 
        {\it Statistical science}, {\bf 7}(4), 457-472.

% \item[] Geyer,  C.J. (1992).  Practical  Markov  chain  Monte  Carlo  (with  discussion).
%      {\it Statistical Science}, {\bf 7}, 473-511.

\item[] Hastie, T., Tibshirani, R. and Friedman, J. (2009). {\it The Elements of Statistical Learning} (2nd edition). 
        Springer. 

\item[] He, S. (2011). Extension of SPACE: R Package `SpaceExt'. Downloadable at \\
     http://cran.r-project.org/web/packages/spaceExt. 

\item[] He, Y. and Liu, C. (2012). The dynamic 'expectation-conditional maximization either' algorithm. 
  \JRSSB, {\bf 74}, 313-336. 

% \item[] Jennrich, R.I. (1969). Asymptotic properties of non-linear least squares estimators. {\it The Annals of Mathematical
%         Statistics}, {\bf 40}, 633-643.

\item[] Johnson, V.E. and Rossell, D. (2012). Bayesian model selection in high-dimensional settings. 
        \JASA, {\bf 107}, 649-660.

\item[] Khalili, A. and Chen, J. (2007).  Variable selection in finite mixture of regression models. \JASA, 
  {\bf 102}, 1025-1038.

\item[] Khamsi, M.A. and Kirk, W.A. (2000). {\it An Introduction to Metric Spaces and Fixed Point Theory}. 
        Wiley. 

% \item[] Kolaczyk, E.D. (2009). {\it Statistical Analysis of Network Data: Methods and Models}.
%        Springer.

% \item[] Lauritzen, S. (1996). {\it Graphical Models}. Oxford: Oxford University Press.

\item[] Liang, F., Song, Q. and Qiu, P. (2015). An equivalent measure of partial correlation coefficients 
  for high-dimensional Gaussian graphical models. \JASA, {\bf 110}, 1248-1265. 

% \item[] Liang, F., Song, Q. and Yu, K. (2013). Bayesian subset modeling for high dimensional GLMs.
%        \JASA, {\bf 108}, 589-606.

\item[] Liang, F. and Zhang, J. (2008). Estimating the false discovery rate using the
  stochastic approximation algorithm.  {\it Biometrika}, {\bf 95}, 961-977.

\item[] Liu, C. and Rubin, D.B. (1994). The ECME algorithm: a simple extension of EM and ECM with faster monotone 
        convergence. {\it Biometrika}, {\bf 81}, 633-648. 

\item[] Liu, C., Rubin, D.B., and Wu, Y.N. (1998). Parameter expansion to accelerate EM: the PX-EM algorithm. 
        {\it Biometrika}, {\bf 85}, 755-770. 

\item[] Long, Q. and Johnson, B.A. (2015). Variable selection in the presence of missing data: resampling 
        and imputation. {\it Biostatistics}, {\bf 16}, 596-610.

\item[] Mazumder, R. and Hastie, T. (2012). The graphical lasso: New insights and alternatives.
  {\it Electronic Journal of Statistics}, {\bf 6}, 2125-2149.

\item[] Mazumder, R., Hastie, T. and Tibshirani, R. (2010). Spectral regularization algorithms for learning 
        large incomplete matrices. \JMLR, {\bf 99}, 2287-2322. 

\item[] McLachlan, G.J. and Krishnan, T. (2008). {\it The EM Algorithm and Extensions} (2nd edition),
        Wiley.

\item[] Meinshausen, N. and B\"uhlmann, P. (2006). High-dimensional graphs and variable selection
     with the Lasso. {\it Annals of Statistics}, {\bf 34}, 1436-1462.

\item[] Meinshausen, N. and B\"uhlmann, P. (2010). Stability selection. \JRSSB, {\bf 72}, 417-473.

 \item[] Meng, X.-L. and Rubin, D.B. (1993). Maximum likelihood estimation via the ECM algorithm: a general framework.
         {\it Biometrika}, {\bf 80}, 267-278.

\item[] Meyn, S. and Tweedie, R.L. (2009). {\it Markov Chains and Stochastic Stability} (2nd Edition). 
          Cambridge University Press. 

%\item[] Newey, W.K. (1991). Uniform convergence in probability and stochastic equicontinuity. 
%       {\it Econometrica}, {\bf 59}, 1161-1167. 

\item[] Nielsen, S.F. (2000). The stochastic EM algorithm: Estimation and asymptotic results. 
        {\it Bernoulli}, {\bf 6}, 457-489. 

\item[] Oba, S., Sato, M., Takemasa, I., Monden, M., Matsubara, K.-I., and Ishii, S. (2003). 
        A Bayesian missing value estimation method for gene expression profile data. 
        {\it Bioinformatics}, {\bf 19}, 2088-2096. 

\item[] Ouyang, M., Welsh, W.J., Georgopoulos, P. (2004). Gaussian mixture clustering and imputation of microarray data.
        {\it Bioinformatics}, {\bf 20}, 917-923.

%\item[] P\"otscher, B.M. and Prucha, I.R. (1994). Generic uniform convergence and equicontinuity concepts for random 
%        functions. {\it Journal of Econometrics}, {\bf 60}, 23-63. 

\item[] Raskutti, G., Wainwright, M.J. and Yu, B. (2011). Minimax rates of estimation for high-dimensional 
        linear regression over $l_q$-balls. {\it IEEE Transactions on Information Theory}, {\bf 57}(10), 6976-6994.

\item[] Scheetz, T.E. Kim, K.-Y., Swiderski, R.E., Philp, A.R., et al. (2006). 
        Regulation of gene expression in the mammalian eye and its relevance to eye disease. \PNAS, {\bf 103}, 14429-14434.

\item[] Song, Q. and Liang, F. (2015a). A Split-and-Merge Bayesian Variable Selection Approach for Ultra-high
 dimensional Regression. \JRSSB, 77(5), 947-972.

\item[] Song, Q. and Liang, F. (2015b). High Dimensional Variable Selection with Reciprocal $L_1$-Regularization.
           \JASA, {\bf 110}, 1607-1620.

\item[] St\"adler, N. and B\"uhlmann, P. (2012). Missing values: sparse inverse covariance estimation and an 
    extension to sparse regression. {\it Statistics and Computing}, {\bf 22}, 219-235. 

\item[] St\"adler, N., Stekhoven, D.J. and  B\"uhlmann, P. (2014). Pattern alternating maximization 
      algorithm for missing data in high-dimensional problems. 
      {\it Journal of Machine Learning Research}, {\bf 15}, 1903-1928. 

\item[] Storey, J.D. (2002). A direct approach to false discovery rates. \JRSSB, {\bf 64}, 479-498.

\item[] Tanner, M.A. and Wong, W.H. (1987). The calculation of posterior distributions by 
        data augmentation (with discussion). \JASA, {\bf 82}, 528-540. 

\item[] Tibshirani, R. (1996). Regression shrinkage and selection via the LASSO. \JRSSB, {\bf 58}, 267-288.

\item[] Troyamskaya, O., Cantor, M., Sherlock, G., Brown, P., Hastie, T., Tibshirani, R., Botstein, D.,
        and Altman, R. (2001). Missing value estimation methods for DNA microarrays.
        {\it Bioinformatics}, {\bf 17}, 520-525. 

\item[] Tseng, P. (2001). Convergence of a block coordinate descent method for nondifferentiable minimization. 
        {\it Journal of Optimization Theory and Application}, {\bf 109}(3), 475-494. 

\item[] Tseng, P. and Yun, S. (2009). A coordinate gradient descent method for nonsmooth separable minimization. 
        {\bf 117}(1-2), 387-423. 

\item[] van de Geer, S., B\"uhlmann, P., Ritov, Y., and Dezeure, R. (2014). On asymptotically 
        optimal confidence regions and tests for high-dimensional models. \ANNALS, {\bf 42}, 1166-1202. 

\item[] van der Vaart, A.W. and Wellner, J.A. (1996). {\it Weak Convergence and Empirical Processes}. 
        Springer. 

\item[] Vershynin, R. (2015). Estimation in high dimensions: A geometric perspective.
        In: Pfander G. (eds) {\it Sampling Theory, a Renaissance}, Birkhäuser, Cham, pp.3-66.

%  \item[] White, H. (1980). Nonlinear regression on cross-section data. {\it Econometrica}, {\bf 48}, 721-726.
  
% \item[] White, H. (1981). Consequences and detection of misspecified nonlinear regression models. 
%         \JASA, {\bf 76}, 419-433. 
 
% \item[] White, H. (1982). Maximum likelihood estimation of misspecified models. {\it Econometrica}, {\bf 50}, 1-25. 

\item[] Wei, G.C.G. and Tanner, M.A. (1990). A Monte Carlo implementation of the EM algorithm and the poor
        man's data augmentation algorithms. \JASA, {\bf 85}, 699-704. 

 \item[] Wu, C.F.J. (1983). On the convergence properties of the EM algorithm. \ANNALS, {\bf 11}, 95-103. 

\item[] Yu, G. and Liu, Y. (2016). Sparse regression incorporating graphical structure among predictors. 
        \JASA, {\bf 111}, 707-720.

\item[] Yuan, M. and Lin, Y. (2007). Model selection and estimation in the Gaussian graphical model.
       {\it Biometrika}, {\bf 94}, 19-35.

\item[] Zhang, C.-H. (2010). Nearly unbiased variable selection under minimax concave penalty. 
   \ANNALS, {\bf 38}, 894-942. 

\item[] Zhang, C.-H. and Zhang, S.S. (2014). Confidence intervals for low dimensional parameters 
        in high dimensional linear models. \JRSSB, {\bf 76}, 217-242. 

\item[] Zhao, Y. and Long, Q. (2013). Multiple imputation in the presence of high-dimensional data. 
        {\it Statistical Methods in Medical Research}, 1-15. 

%\item[] Zou, H. and Hastie, T. (2005). Regularization and variable selection via the elastic net. 
%        \JRSSB, {\bf 67}, 301-320.
 
\end{description}

 \end{document}